\documentclass[final,5p,times,twocolumn,authoryear]{elsarticle}

\usepackage[T1]{fontenc}

\usepackage{amssymb}
\usepackage{color}
\usepackage{geometry}
\usepackage{ulem}
\usepackage{soul}
\usepackage{caption}
\usepackage{amsmath,amsthm}
\newtheorem{lem}{Lemma}

\definecolor{olive}{RGB}{33,126,38}
\newcommand{\R}{\mathbb{R}}
\newcommand{\N}{\mathbb{N}}

\begin{document}

\begin{frontmatter}


\author{Charline Smadi \fnref{label1}}
\author{H\'el\`ene Leman \fnref{label2}}
\author{Violaine Llaurens \corref{co} \fnref{label3}}

\ead{violaine.llaurens@mnhn.fr}

\fntext[label1]{IRSTEA UR LISC, Laboratoire d'ing\'enierie des Syst\`emes Complexes, 9 avenue Blaise-Pascal CS
20085, 63178 Aubi\`ere, France and Complex Systems Institute of Paris Île-de-France, 113 rue Nationale, Paris, France}
\fntext[label2]{CIMAT, De Jalisco S-N, Valenciana, 36240 Guanajuato, Gto., Mexico}
\fntext[label3]{Institut de Syst\'ematique, Evolution et Biodiversit\'e UMR 7205 CNRS/MNHN/UPMC/EPHE, Museum National d'Histoire Naturelle, 
Paris, France}

 \cortext[co]{Corresponding author}

\title{Looking for the right mate in diploid species: how dominance relationships affect population differentiation in sexual trait?}


\author{}

\address{}

\begin{abstract}
Divergence between populations for a given trait can be driven by sexual selection, interacting with migration behaviour. Mating preference for different phenotypes 
may indeed lead to specific migration behaviour, with departures from populations where the preferred trait is rare. 
Such preference can then trigger the emergence and persistence of differentiated populations, even without any local adaptation. However the genetic architecture 
underlying the trait targeted by mating preference may have 
a profound impact on population divergence. In particular, dominance between alleles encoding for divergent phenotypes can interfere in such differentiation process. 
Using a diploid model of trait determining both mating 
success and migration rate, we explored differentiation between two connected populations, assuming either co-dominance or strict dominance between alleles. The model 
assumes that individuals prefer mating with partners displaying the same phenotype and therefore tend to move to the other population when their own phenotype 
is rare. We show that the emergence of differentiated populations in this diploid moded is limited as compared to results obtained with the same model assuming 
haploidy. When assuming co-dominance, differentiation 
arises only when migration is limited as compared to preference. Such differentiation is less dependent on migration when assuming strict dominance between 
haplotypes. 
Dominant alleles frequently invade populations because their phenotype is more frequently expressed, resulting in higher local mating success and rapid decrease in 
migration. However, depending on the initial distribution of alleles, this advantage associated with dominance (i.e. Haldane's sieve) may lead  to fixation of the 
dominant allele throughout both populations. Depending on the initial distribution of heterozygotes in the two populations, persistence of polymorphisms within 
populations can also occur because heterozygotes displaying the predominant phenotype benefit from high mating success. Altogether, our results highlight 
that heterozygotes' behaviour 
has a strong impact on population differentiation and stress out the need of diploid models of differentiation and speciation driven by sexual selection.

\end{abstract}


\end{frontmatter}

\section{Introduction}

Understanding processes leading to biological diversification is a central question in evolutionary biology. Traits may diverge neutrally because of geographic 
barriers limiting gene flow, or simple isolation by distance due to limited dispersal, resulting in genetic and phenotypic differentiation
(\cite{lande1980genetic,slatkin1987gene}). 
Mating preferences may also drive divergence in targeted traits, with assortative mating promoting local fixation of the most abundant phenotype: 
common phenotypes indeed benefit from greater mating success leading to positive frequency dependent selection at local scale. Migration behaviours can then be 
affected by such mating preference, because mate searching can stimulate dispersal when the preferred trait is locally rare. Altogether, mate preference has been 
shown to have a great impact on both local polymorphism (\cite{PayneKrakauer1997}) and spatial 
structure (\cite{m2010assortative}) of targeted traits. Nevertheless the genetic architecture of the trait under sexual selection may also influence the evolution of such population differentiation. Linkage disequilibrium 
between loci controlling the adaptive trait and preference trait is known to favour the divergence in a Fisherian run-away process (\cite{fisher1930genetical}). 
Dominance relationships among differentiated alleles may also greatly influence the spatial distribution of different phenotypes 
(\cite{pannell2005haldane}). The effective migration of 
advantageous alleles is indeed favoured when they are dominant: immigrant alleles entering a new population will mostly occur at heterozygous state so that they will be picked up by 
positive selection only if they are expressed. Recessive adaptive alleles, scarcely expressed, are more likely to be lost by genetic drift because of their 
neutrality at heterozygous state. This Haldane's sieve effect (\cite{haldane1927mathematical}) predicts a greater effective migration of dominant adaptive alleles as compared to 
recessive ones. For example, in the polymorphic locus of sporophytic self-incompatibility where rare alleles benefit from increased reproductive success, migration 
of dominant alleles has been shown to be more effective than migration of recessive ones (\cite{schierup1997evolutionary}). 
Dominance among alleles 
may thus play an important role on the dynamics of spatial differentiation of traits under sexual selection. 
Here we investigate the influence of dominance 
on spatial differentiation of a trait determining both mating and migration behaviours.
Our model is an extension of a previously described model assuming haploid individuals 
(\cite{coron2016stochastic}). 
Two populations, linked by migration, are assumed. A mating trait, encoded by one locus has two consequences: 
(1) encountering pairs mate more often when both individuals display the same phenotype for this trait and (2) 
the migration rate of an individual is proportional to the frequency of the other phenotype in its patch:
individuals are more prone to move if they have difficulties to find a suitable mate in their patch. This hypothesis is relevant for all organisms with active mate searching  
($e.g.$ patrolling behaviours in butterflies (\cite{jugovic2017movement}) or mate sampling in Bowerbirds (\cite{uy2001complex})), but also in 
organisms where gametes are involved in 
dispersal ($e.g.$ sea urchins (\cite{crimaldi2012role}) or plants (\cite{millar2014extensive})), and may travel large distance before encountering a suitable mate. 
This original hypothesis linking mating preference and migration contrasts with assumptions found in classical models of speciation (see (\cite{gavrilets2014review}) 
for a review)  
where the preference traits are generally not directly linked to migration behaviour.
In the haploid model studied in (\cite{coron2016stochastic}), this preference behaviour can lead to spatial 
differentiation of the trait between the two populations with fixation of different phenotypes in the two populations despite migration, without any local adaptation.
More precisely, if at initial state one phenotype is predominant in one patch and the other phenotype predominant in the other patch, 
then there is a spatial differentiation regardless of migration rate value.
Given the importance of dominance on migration of alleles, here we
extend this model to a diploid case with dominance relationships between alleles, and  explore how dominance may modulate the dynamics of spatial differentiation 
on this mating trait.

\section{Materials and methods}

In this section we provide a general description of the model generating
population dynamics. We also explicit mathematical methods used to derive analytical results as well as conditions used for numerical simulations.
More details can be found in the Supplementary Materials.

We consider a population of hermaphroditic diploid individuals characterized by 
(1) a single phenotype controlled by their genotype at one bi-allelic locus ($A$ and $a$), 
and (2) their position on a space divided in two patches ($1$ and $2$). 
The number of individuals $AA$, $Aa$ and $aa$ in the two patches follow a dynamical system, which can be obtained as a limit of a 
stochastic multi-type birth and death 
process with competition in continuous time ((\cite{coron2016stochastic}) and see \ref{def_model}).
In particular, populations sizes in the patches are varying and generations are overlapping.

The phenotype of individuals influences both (1) their mating and (2) migration behaviours. 
(1) Individuals having 
the same phenotype have a higher probability to mate (see (\cite{jiang2013assortative}) for a recent review of the mechanisms of assortative 
matings in animals). 
(2) This mate preference also influences migration from one patch to another: individuals carrying a phenotype at low frequency within patch have a greater migration because we assumed migration to be promoted by the local lack of suitable mates. Therefore, migration rate depends on the frequency of individuals carrying a different phenotype within their patch. 
Examples of animals migrating to find suitable mates are well documented (\cite{Schwagmeyer1988, Honeretal2007}).
A migration mechanism similar to the one presented here has been studied in (\cite{coron2016stochastic}) and in (\cite{PayneKrakauer1997}) in a continuous space model.

Five parameters are needed to describe the two-populations dynamics:
\begin{itemize}
 \item $b$ is the minimal individual birth rate. It corresponds to the rate at which an individual gives birth if there is no individual with the same 
 phenotype in its patch.
 \item $\beta \geq 1$ is the sexual preference. 
 Individuals encounter randomly, and 
 two individuals with 
 similar phenotypes have a higher probability to mate (see Section \ref{section_dominance} for details).
 \item $p$ describes the individuals' ability to migrate. The effective migration rate of an individual is the product of $p$ and of the proportion 
 of 'not preferred' individuals (see Section \ref{section_model} and \ref{def_model} for details). 
 \item $d$ is the individual natural death rate.
 \item $c$ represents the competition for food or space exerted by an individual on another one of the same patch. An additionnal 
 individual death rate results from the total competition exerted on this individual.
\end{itemize}
There is no geographically variable selection: the parameters $b$, $p$, $d$, and $c$ do not depend on the patch, thus the two patches have the same ecological characteristics.
However, depending on the patch's composition, the individual's
behaviour in terms of reproduction and migration may vary in time and space.

\subsection{Dominance} \label{section_dominance}

We assume Mendelian segregation so that any parent transmits each of its alleles with probability $1/2$.
To study the effect of dominance on population differentiation, we contrast two opposite scenarios: complete co-dominance and complete dominance.
Dominance observed in natural populations can differ from these extreme cases, however by studying the limits of the dominance continuum we cover the possible 
population dynamics.
In the complete dominance scenario, individuals with genotypes $AA$ and $Aa$ have the same phenotype, A, whereas individuals with genotype $aa$ have the 
phenotype a. In the co-dominance scenario, heterozygotes express an intermediate phenotype between either homozygotes and we thus consider two possibilities. 
(1) Preference expressed towards heterozygotes will be intermediate between assortative and disassortative mating (COD 1) ($\beta$ for pairs $(AA,AA)$, 
$(aa,aa)$ and $(Aa,Aa)$, $(\beta+1)/2$ for pairs $(AA,Aa)$ and $(aa,Aa)$, and $1$ for pairs $(AA,aa)$) and migration rate varies accordingly because the decision to leave the patch depends on the lack of 
preferred partners in the patch. (2) Heterozygotes express no preference towards any partners (COD2), the preference parameter is thus $\beta$ for pairs $(AA,AA)$ and $(aa,aa)$, and $1$ for all other pairs. Because of this lack of preference, heterozygotes have no reason to look for suitable mate and thus do not migrate.

\subsection{Model} \label{section_model}

Our model can be precisely described as follows. We denote the number of 
individuals of genotype  $\mathfrak{g}$ in the patch $i$ at time $t$ for any $ \mathfrak{g}\in \{AA,Aa,aa\}$, $i\in\{1,2\}$ and $t\geq 0$ by $z_{ \mathfrak{g},i}(t)$. 
Moreover we denote the total number of individuals in the patch $i$ at time $t$ by 
$$N_i(t):=z_{AA,i}(t)+z_{Aa,i}(t)+z_{aa,i}(t).$$
 The parameter $\beta$ measures the strength of the sexual preference: $\beta=1$ means no preference, and a large $\beta$ indicates a strong preference. 
Since we are only investigating assortative mating here, $\beta$ will always be larger than one.
 We denote by $p_\beta( \mathfrak{g}, \mathfrak{g}')$ the preference between two individuals with genotypes $ \mathfrak{g}$ and $ \mathfrak{g}'$, 
respectively. They differ according to the model considered (see Table \ref{table_pref}) and always belong to $[1,\beta]$.
Finally, the parameter $p$ describes the maximum migration rate of an individual. The migration rate of an individual is the product of this parameter $p$, 
of the proportion of 'non-suitable' mates in its patch, and of the function of preferences $p_\beta(.,.)$ between genotypes.\\

The dynamical systems governing the evolution of the population was obtained as follows.
Individuals are assumed hermaphroditic and diploid. At a rate $B>b\beta$, they will reproduce as female and look for a mate. They will encounter uniformly an 
individual reproducing as male from the population.
The probability that they actually mate and produce an offspring is 
$bp_\beta( \mathfrak{g}, \mathfrak{g}')/B$, where the functions $p_\beta$ are detailled in Table~\ref{table_pref}.

In the patch $1$, $AA$ individuals reproducing as female will generate:
\begin{itemize}
 \item an offspring of type $AA$ in the patch $1$ at a rate 
 $$ \frac{Bz_{AA,1}}{z_{AA,1}+z_{Aa,1}+z_{aa,1}}\left( \frac{b\beta}{B}z_{AA,1}+ \frac{bp_\beta(AA,Aa)}{B}\frac{z_{Aa,1}}{2} \right). $$
 Indeed, as we consider a Mendelian segregation, parents of genotypes $AA$ and $Aa$ generate an offspring of type $AA$ with a probability $1/2$ and 
 an offspring of type $Aa$ with a probability $1/2$.
 \item an offspring of type $Aa$ in the patch $1$ at a rate 
\end{itemize}
 $$ \frac{Bz_{AA,1}}{z_{AA,1}+z_{Aa,1}+z_{aa,1}}\left( \frac{b}{B}z_{aa,1}+ \frac{bp_\beta(AA,Aa)}{B}\frac{z_{Aa,1}}{2} \right). $$
Since individuals are hermaphrodites and can reproduce through both male and female pathways, we finally get the following equation for the population 
dynamics on a patch $i \in \{1,2\}$ (where $j$ is the complement of $i$ in $\{1,2\}$):

\begin{equation}
\label{eq_model}
\left\{
 \begin{aligned}
 \dot z_{AA,i}=&\frac{b }{N_i}\left( \beta z_{AA,i}^2+p_\beta(AA,Aa)z_{AA,i}z_{Aa,i}+\frac{p_\beta(Aa,Aa)}{4} z_{Aa,i}^2 \right) \\&-(d+cN_i)z_{AA,i}\\
   &+p \sum_{ \mathfrak{g} \in  \mathcal{G}} \left( \frac{p_\beta(AA,AA)-p_\beta(AA,\mathfrak{g})}{\beta-1} \right) 
   \left(\frac{z_{\mathfrak{g},j}}{N_j}z_{AA,j}-\frac{z_{\mathfrak{g},i}}{N_i}z_{AA,i}\right)\\
 \dot z_{Aa,i}=&\frac{b}{N_i}\left( \frac{p_\beta(Aa,Aa)}{2} z_{Aa,i}^2+p_\beta(AA,Aa)z_{Aa,i}z_{AA,i}\right.\\&+p_\beta(aa,Aa)z_{Aa,i}z_{aa,i} 
  +2 z_{AA,i}z_{aa,i} \Big)-(d+cN_i)z_{Aa,i} \\
 &+p \sum_{ \mathfrak{g} \in  \mathcal{G}} \left( \frac{p_\beta(Aa,Aa)-p_\beta(Aa,\mathfrak{g})}{\beta-1} \right) 
   \left(\frac{z_{\mathfrak{g},j}}{N_j}z_{Aa,j}-\frac{z_{\mathfrak{g},i}}{N_i}z_{Aa,i}\right)\\
 \dot z_{aa,i}=&\frac{b }{N_i}\left( \beta z_{aa,i}^2+p_\beta(aa,Aa)z_{aa,i}z_{Aa,i}\right. \\
 &\left. +\frac{p_\beta(Aa,Aa)}{4} z_{Aa,i}^2 \right) -(d+cN_i)z_{aa,i}
  \\
 &+p \sum_{ \mathfrak{g} \in  \mathcal{G}} \left( \frac{p_\beta(aa,aa)-p_\beta(aa,\mathfrak{g})}{\beta-1} \right) 
   \left(\frac{z_{\mathfrak{g},j}}{N_j}z_{aa,j}-\frac{z_{\mathfrak{g},i}}{N_i}z_{aa,i}\right)
\end{aligned}
\right. 
\end{equation}
where $\mathcal{G}:=\{AA,Aa,aa\}$.
Hence the closer to $\beta$ the preference between two genotypes is, the smaller the individuals' migration rates of individuals with one of these 
genotypes will be
in presence of individuals with the other genotype.
Note that individual migration rates are always between $0$ and $p$.
The equations followed by the dynamics under the three dominance hypotheses are then obtained by replacing the preference functions $p_\beta$ 
by their value, as summarized in Table \ref{table_pref} (see \ref{def_model} for a presentation of the full equations in the different models).

\begin{center}
\captionof{table}{Preference functions assuming dominance hypothesis (COD1), (COD2) and (DOM).}\label{table_pref}
\begin{tabular}{|c|c|c|c|} 
\hline
&(COD1)&(COD2)&(DOM)\\
\hline
$p_\beta(AA,AA)$&$\beta$&$\beta$&$\beta$\\
\hline
$p_\beta(AA,Aa)$&$(\beta+1)/2$&$1$ &$\beta$\\
\hline
$p_\beta(Aa,Aa)$&$\beta$& $1$&$\beta$\\
\hline
$p_\beta(aa,Aa)$&$(\beta+1)/2$&$1$&$1$\\
\hline
$p_\beta(aa,aa)$& $\beta$&$\beta$& $\beta$\\
\hline
$p_\beta(AA,aa)$& $1$&$1$& $1$\\
\hline
\end{tabular}
\end{center}

\subsection{Mathematical analysis}

We carry out a mathematical analysis of the dynamical systems governing the population evolution, fully described in Supplementary Materials. 
Here, we only show some fixed points of the systems, their stability, and obtain some convergence results in the case without migration. 
We used the theory of dynamical systems (in particular Lyapunov functions and the Local Center Manifold Theorem) as well as the 
theory of polynomial functions.

\subsection{Simulations}

To illustrate some dynamics of the dynamical systems governing the evolution of the population sizes, we used the software Mathematica.

We also performed numerical simulations of the dynamical systems presented in \ref{def_model}. 
We investigated the solution of the dynamical systems for $b=2$, $d=1$, $c=0.5$ and 
different values of migration rate $p$ and of preference coefficient $\beta$. For each value of $\beta$ and $p$, we solved the dynamical systems 
for $10\, 000$ different initial conditions chosen as follows. First, we set the size of the population in patch $1$: we considered $100$ values of sizes regularly distributed 
between $1$ and $2*(2*b-d)/c$. Then, we set the size of the population in patch $2$ such that the difference between the two sizes was $0.01$.
This allowed us to reduce the number of parameters explored without changing the results.
Finally, for each 
pair of sizes, we examined $100$ initial conditions randomly chosen. For each patch $i$, we set uniformly at random the proportion $p_{A,i}$ of allele 
$A$ using a uniform random variable between $0.5$ and $1$ in the case where the majority allele was $A$ or between $0$ and $0.5$ in the other case. The proportion 
of $Aa$-individuals in patch $i$ was fixed randomly using a uniform random variable between $0$ and $2*min(p_{A,i},1-p_{A,i})$. The proportion of $AA$-individuals 
and $aa$-individuals in each patch can be easily deduced.\\
\indent For each initial condition, we numerically solved the three dynamical systems using a finite difference method. 
We used a discretization time step $h$ equals to $0.005$. We assumed that a 
stationary state was reached as soon as the norm of the difference of the solution between two time steps was lower than $\varepsilon=10^{-6}$. 
Different values of $h$ and $\varepsilon$ were tested and the chosen pair of values provided the best trade-off between algorithm rapidity and result accuracy. Once the stationary state was found, we considered that the final population in a 
patch was monomorphic when the proportion of one of the alleles was larger than 99\%. Otherwise, it was considered as polymorphic equilibrium.

\section{Results}

To highlight the influence of dominance on the spatial dynamics of the trait, we first contrast our findings with  
results obtained in a previous study (\cite{coron2016stochastic}) where individuals were haploids.

\subsection{Haploid model}

In the haploid model, preference functions were $p_\beta(A,A)=p_\beta(a,a)=\beta$ and $p_\beta(A,a)=1$.
The system may converge to two different types of equilibria depending on initial conditions. 
Equilibria can be expressed using $\zeta$, the equilibrium population size in a patch when there is only one type of 
individuals ($A$ or $a$),
\begin{equation} \label{def_zeta} \zeta:= \frac{\beta b-d}{c}. \end{equation} 
Let us denote by $z_{\alpha,i}(t), \alpha \in \{A,a\}, i \in \{1,2\}, t \geq 0$ the $\alpha$ population size in the patch $i$ at time $t$. Then
\begin{itemize}
 \item[$\bullet$] If $z_{A,i}(0)>z_{a,i}(0), i = 1,2$, the population sizes converge to $(z_{A,1},z_{a,1},z_{A,2},z_{a,2})=(\zeta,0,\zeta,0)$.
 \item[$\bullet$] If $z_{A,1}(0)>z_{a,1}(0)$ and $z_{A,2}(0)<z_{a,2}(0)$ the population sizes converge to $(z_{A,1},z_{a,1},z_{A,2},z_{a,2})=(\zeta,0,0,\zeta)$.
\end{itemize}
The same conclusions hold by symmetry if we replace $A$ by $a$ in the first line and $1$ by $2$ in the second one.
Hence when initial conditions are asymmetrical (more $A$ individuals in one patch and more $a$ individuals in the other patch) and when $p$ 
is positive, regardless of its value, it is enough to entail the end of gene flux between the two populations and 
to generate two differentiated populations.

\subsection{Diploid models with co-dominance}

We then compared the results obtained in haploid populations with the diploid model assuming co-dominance between alleles, 
\textit{i.e.} when heterozygotes express an intermediate phenotype between either homozygotes. 
This assumption exhibits high similarity with the haploid case, although the behaviour of heterozygous individuals displaying intermediate phenotypes might 
influence model outputs. We contrasted two hypotheses:
 (1) mating success between homozygotes and heterozygotes were half less 
than between the same genotypes (COD1), 
(2) heterozygotes had no preferences and were not preferred by homozygotes (COD2) (preference parameter $1$ for any encountering pair with an $Aa$ individual). 
We then investigated the equilibrium reached in both populations using a mathematical analysis and
simulations assuming different preference coefficients ($\beta$) and migration rates ($p$).

\subsubsection{System without migration} \label{sectionnomig}

When there is no migration ($p=0$) we are able to give necessary and sufficient conditions on the initial 
numbers of individuals with genotypes $AA$, $Aa$ and $aa$ in both patches for the system to converge to the different fixed points, 
in both codominant models.
Recall that $z_{\alpha \alpha',i}(t)$ denotes the number of individuals with genotype $\alpha \alpha'$ ($AA$, $Aa$ or $aa$) in the patch $i$ ($1$ or $2$)
at time $t$.
As there is no migration, it is enough to consider the patch $1$.
Under hypothesis (COD1) the system follows the equations:
\begin{equation*}
\left\{
 \begin{aligned}
& \dot z_{AA}=\frac{b }{N}\left( \beta z_{AA}^2+\frac{\beta+1}{2}z_{AA}z_{Aa}+\frac{\beta}{4} z_{Aa}^2 \right) -(d+cN)z_{AA}\\
& \dot z_{Aa}=\frac{b}{N}\left( \frac{\beta}{2} z_{Aa}^2+\frac{\beta+1}{2}z_{Aa}(z_{AA}+z_{aa})+2 z_{AA}z_{aa} \right)-(d+cN)z_{Aa}\\
& \dot z_{aa}=\frac{b }{N}\left( \beta z_{aa}^2+\frac{\beta+1}{2}z_{aa}z_{Aa}+\frac{\beta}{4} z_{Aa}^2 \right)-(d+cN)z_{aa},
 \end{aligned}
\right.
\end{equation*}
where $N=z_{AA}+z_{Aa}+z_{aa}$. This system
admits two stable fixed points, $(z_{AA,1}=\zeta,z_{Aa,1}=0,z_{aa,1}=0)$ (fixation of allele $A$) and 
$(z_{AA,1}=0,z_{Aa,1}=0,z_{aa,1}=\zeta)$ (fixation of allele $a$), where we recall that $\zeta$ has been defined in \eqref{def_zeta}, 
and one unstable fixed point with persistence of all three genotypes 
(see \ref{cod1_sans_mig}). If $z_{AA,1}(0)>z_{aa,1}(0)$, allele $A$ gets fixed, and the numbers of individuals converge to the stable equilibrium $(\zeta,0,0)$.
If $z_{AA,1}(0)<z_{aa,1}(0)$, the exact same conclusion holds with $a$ replacing $A$.
Finally, if $z_{AA,1}(0)=z_{aa,1}(0)$, 
the system converges to the unstable equilibrium.
Under hypothesis (COD2), the equations driving the population dynamics are:
\begin{equation*}
\left\{
 \begin{aligned}
& \dot z_{AA}=\frac{b }{N}\left( \beta z_{AA}^2+z_{AA}z_{Aa}+\frac{1}{4} z_{Aa}^2 \right) -(d+cN)z_{AA}\\
& \dot z_{Aa}=\frac{b}{N}\left( z_{Aa}^2+z_{Aa}(z_{AA}+z_{aa})+2 z_{AA}z_{aa} \right)-(d+cN)z_{Aa}\\
& \dot z_{aa}=\frac{b }{N}\left( \beta z_{aa}^2+z_{aa}z_{Aa}+\frac{1}{4} z_{Aa}^2 \right)-(d+cN)z_{aa}.
 \end{aligned}
\right.
\end{equation*}
We get the same result, except that the unstable fixed point is different (see \ref{cod2_sans_mig} for details).

\subsubsection{Fixed points of the system with migration}

When there is a migration between the two patches ($p>0$), the dynamics is much more complex due to the increase in dimensionality
and we were unable to obtain convergence results analytically. 
However we were able to describe some of the fixed 
points and determine their stability.
There are four fixed points with monomorphic populations in both patches: 
fixation of $A$ in both patches $(z_{AA,1}=\zeta,z_{Aa,1}=0,z_{aa,1}=0,z_{AA,2}=\zeta,z_{Aa,2}=0,z_{aa,2}=0)$, 
fixation of $a$ in both patches $(0,0,\zeta,0,0,\zeta)$, or fixation of different alleles in the two patches, 
$(\zeta, 0,0,0,0,\zeta)$ or $(0,0,\zeta,\zeta,0,0)$.
For the two codominant models, the first two fixed points are stable for all parameters values, and the two last fixed points are stable if
$p<p_{crit}=b\beta(\beta-1)/2$ and unstable if $p>p_{crit}$.
The fact that $p_{crit}$ is proportional to $b$ is expected, as we could reduce the number of parameters in Equation \eqref{eq_model} by taking 
$\tilde{b}=b/p$, $\tilde{d}=d/p$, $\tilde{c}=c/p$ and $\tilde{p}=1$.
$p_{crit}$ is increasing with $\beta$. This is also expected as a higher $\beta$ means a higher local advantage due 
to sexual preference for the allele in majority in a patch. Hence a higher $p$ is needed to counteract this advantage 
and create a gene flux. The form of $p_{crit}$ as a function of $\beta$ however is difficult to infer, except the fact 
that it becomes null when $\beta=1$. Indeed, in this last case, the model becomes completely neutral and an infinity 
of equilibria are possible (total population size $(b-d)/c$ in each patch and any proportions of $AA$, $Aa$, and $aa$ individuals).
This result contrasts with the haploid case, where the fixed points 
with a genotype in each patch
were stable for all the values of the migration parameter $p$ (see \cite{coron2016stochastic}).
This may be explained by the fact that the migration of heterozygotes has a major impact on the population behaviour, as explicited below using numerical simulations. 
In the figure displaying the behaviour of the model in the codominant case, we have indicated the curve $p_{crit}=b\beta(\beta-1)/2$ (see Fig. \ref{fig1}).

\subsubsection{Conditions for differentiated populations}

\begin{figure}[h]
    \centering
     \includegraphics[width=5cm,height=5cm]{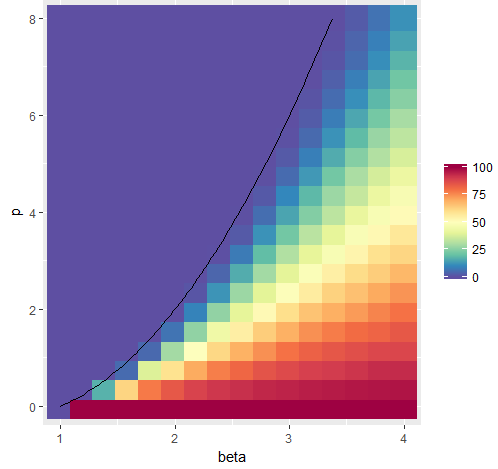}
     \caption{Conditions for the emergence of differentiated populations
 with fixation of allele $A$ in population 1 and $a$ in populations 2, assuming hypothesis (1) regarding codominance (outcomes are very similar under the alternative 
 hypothesis of codominance).
 Initial conditions are asymmetric, with more than $50\%$ of $A$ allele in patch $1$, and 
 more than $50\%$ of $a$ allele in patch $2$.
 Color indicates the percentage of simulations where $A$ get fixed in population $1$ and $a$ in population $2$. 
 Black line shows the limit of the stability of the equilibrium with differentiated populations, $p_{crit}=b\beta(\beta-1)/2$ (see previous section).}
   \label{fig1}
   \end{figure}

Under both assumptions regarding co-dominance, differentiated populations with fixation of different phenotypes in the two populations can emerge only when the 
frequencies of allele $a$ are asymmetrical at initial 
state, with frequency of allele $a$ smaller than 1/2 in one population and larger than 1/2 in the other population (see \ref{app_maj_allele}).
As shown in \ref{cod1_mig} and \ref{cod2_mig}, this equilibrium cannot be reached when $p>b\beta(\beta-1)/2$. However, an initial asymmetry 
and the condition $p<b\beta(\beta-1)/2$ 
are not enough to ensure differentiated populations, as can be seen from Fig. \ref{fig1}.

The number of simulations exhibiting differentiation increases when the preference coefficient $\beta$ increases as expected, 
and decreases when the migration rate $p$ increases.
Indeed when migration increases or preference decreases, this equilibrium becomes unstable (see black curve in Fig.1) and is 
no longer reached.
This contrasts with the dynamics of the model when individuals are haploid, where
differentiated populations emerge as soon as the frequencies of allele $a$ are 
asymmetrical at initial 
state, regardless of the value of the migration strength, $p$, and preference $\beta>1$.
However, the negative effect of migration on population differentiation is 
observed under both assumptions regarding co-dominance, including hypothesis (2) where heterozygotes never migrate.
This indicates that their presence may be enough to maintain the migration of both $A$ and $a$ homozygotes across populations when $p$ is large, 
even if they do not move themselves.

\subsubsection{Fixation of a single phenotype throughout both populations}

\begin{figure}[h]
    \centering
     \includegraphics[width=9cm,height=9cm]{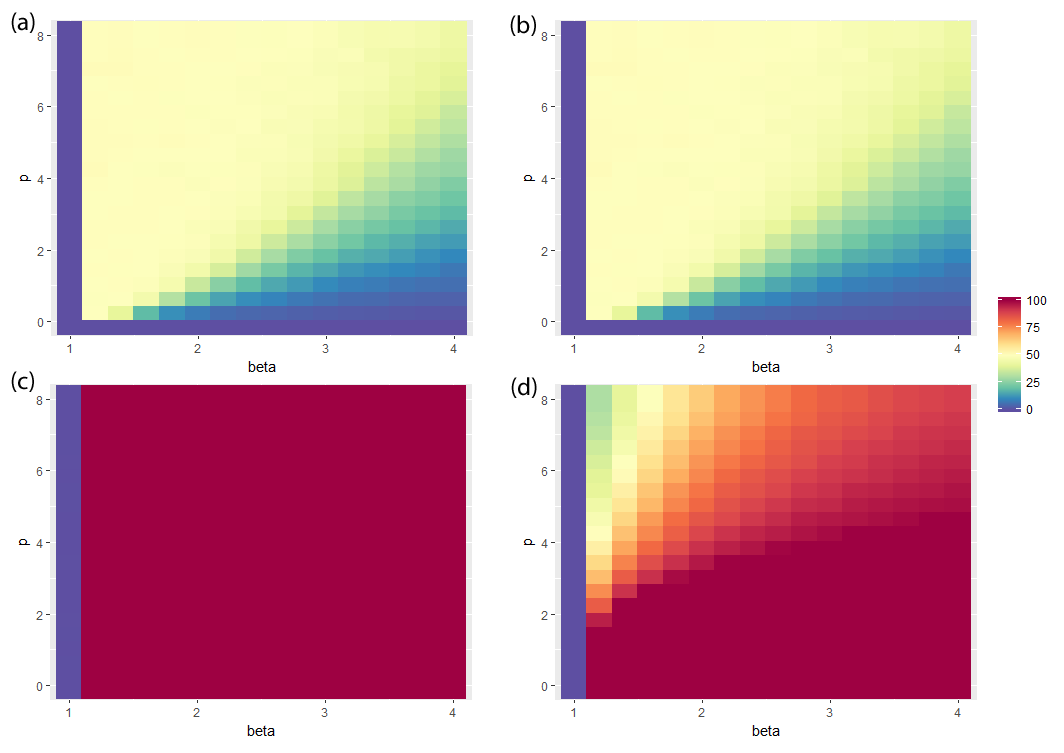}
     \caption{Conditions of fixation of allele $a$ in both populations. Columns
 represent simulations assuming hypotheses (1) and (2) regarding
heterozygotes behaviour. Rows differ in initial conditions. First row:
 Asymmetrical frequency of allele $a$ (more frequent in population 2);
Second row: Frequency of allele $a$ greater than 0.5 in both populations.
The colors indicate the percentage of simulations where $a$ get fixed in both populations.}
   \label{fig2}
\end{figure}

Under hypothesis (COD1), when migration increases, the two populations tend to be more homogeneous, leading to the fixation of a single allele throughout 
both populations, whatever the initial frequencies. 
The identity of the fixed allele depends on the initial frequency because of positive frequency-dependent selection triggered
by homogamy. In cases where initial frequencies of allele $a$ are asymmetrical in the two populations, we indeed observed  
fixation of allele \textit{a} in $50\%$ of simulations 
when migration increases (Fig 2a-b). 
As the system is symmetrical in $a$ and $A$,
in the other half of simulations, the fixation of allele $A$ was observed (data not shown).
When allele $a$ initially predominates in both patches, the fixation of allele $a$ is observed in all simulations for hypothesis (COD1) regarding heterozygote 
behaviour (Fig. 2c). 
When heterozygotes express a preference for themselves and migrate, their migration leads to a fast equalization of 
the numbers of individuals
with the same genotype in the two populations (see Fig. \ref{fig0dyn}) (i.e. $n_{A A,1}$ (resp. $n_{Aa,1}$, $n_{aa,1}$) very close to $n_{A A,2}$ 
(resp. $n_{Aa,1}$, $n_{aa,2}$)).
Once this equalization is reached, the migration does not influence the dynamics since the numbers of emigrants and immigrants of each population are the same. 
Both patches evolve as if they were isolated.
From the study of the system without migration in Section \ref{sectionnomig}, we know that there is no polymorphic stable equilibrium in one isolated patch.
As a consequence, the same allele gets fixed in both patches.

 \begin{figure}[h]
    \centering
     \includegraphics[width=8cm,height=3.5cm]{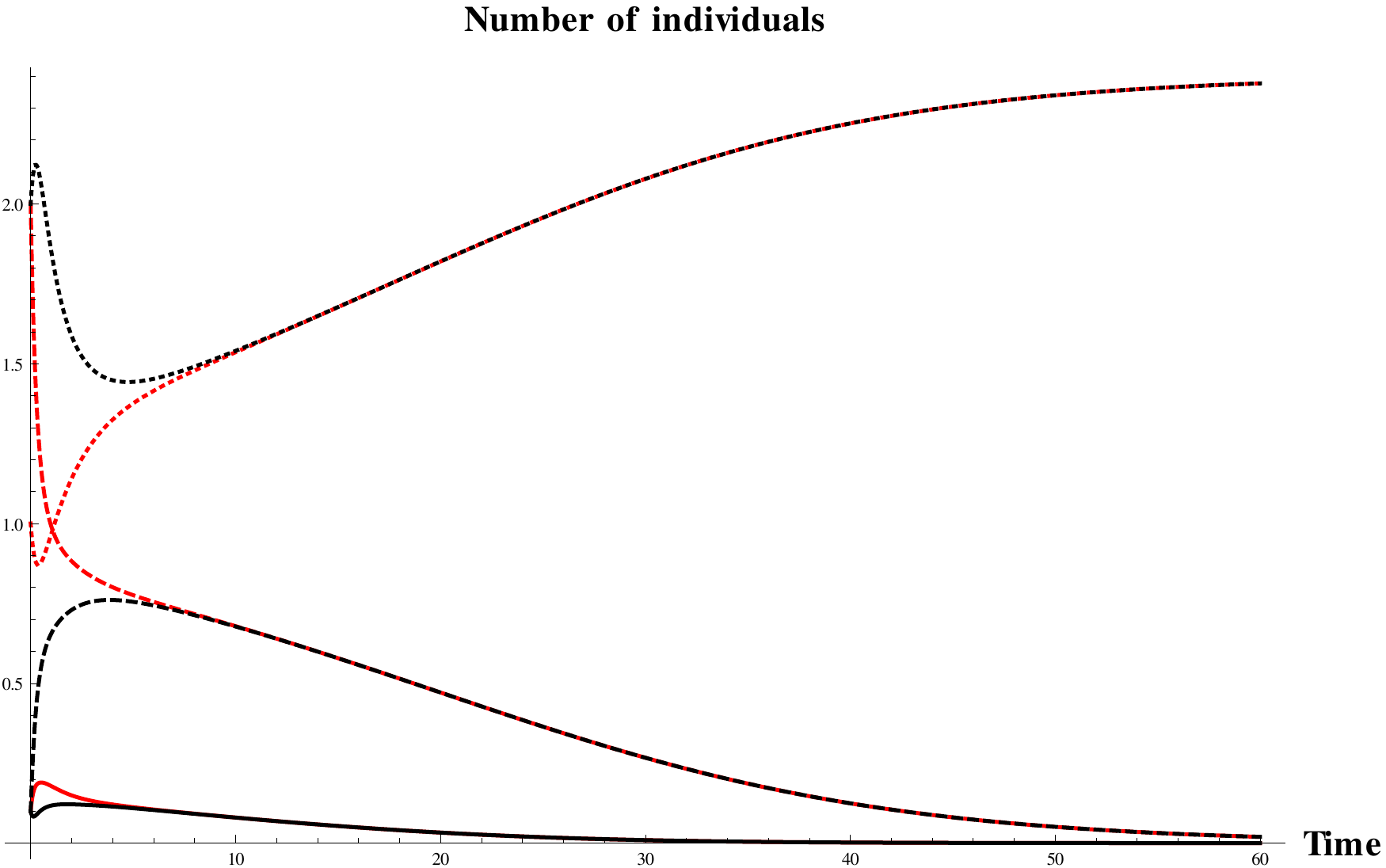} 
 \caption{Number of individuals in both patches under hypothesis (1). $\beta=1.1$ and $p=5$.
     Colors: the dynamics of the number of individuals in the patch $1$ (resp. $2$) are represented in red (resp. black), 
     the dynamics of the number of individuals with genotype $AA$ (resp $Aa$, $aa$) are represented using a full line (resp. dashed line, dotted line) 
     The initial conditions are $z_{AA,1}(0)=z_{AA,2}(0)=z_{Aa,2}(0)=0.1$, $z_{Aa,1}(0)=2$, $z_{aa,1}(0)=1$ and $z_{aa,2}(0)=2$.}
\label{fig0dyn}
     \end{figure}

Assuming no preference and no migration of heterozygotes (hypothesis (COD2)) leads to the overall fixation of one of the two alleles
only when migration is limited or preference 
is high (Fig. 2d). Otherwise in many simulations, polymorphism persists in both populations.

\subsubsection{Polymorphic equilibria}

The dynamics leading to these equilibria are highly non-monotonic and migration persists at equilibrium, which explains the difficulty to study the model 
analytically. The presence of individuals of type $Aa$ maintains migration of $AA$ and $aa$ individuals at equilibrium
 (see Fig. \ref{fig1dyn} and \ref{cam1}).
According to numerical simulations, stable polymorphic equilibria are of the following form:
\begin{enumerate}
 \item[-] one 'large' population (with a size larger than $(\beta b-d)/c$) with homozygote type in large majority, $AA$ or $aa$ respectively, 
 \item[-] one 'small' population (with a size smaller than $(b-d)/c$) with a large proportion of individuals of type $aa$ (resp. $AA$) and $Aa$.
 \end{enumerate}
 \begin{figure}[h]
    \centering
    \includegraphics[width=8cm,height=3.5cm]{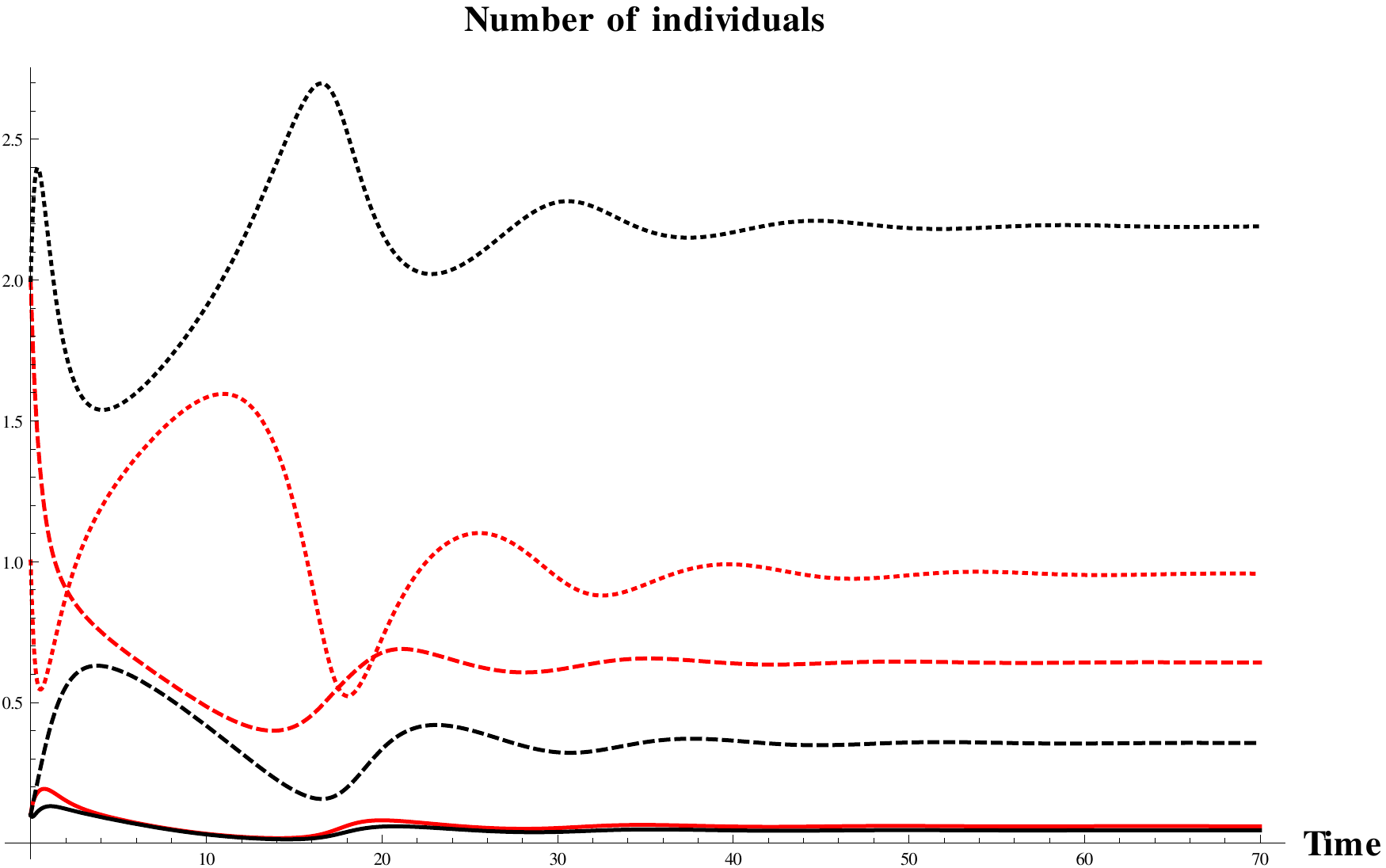}  \hfill \\ \hspace{.5cm}   \includegraphics[width=8cm,height=3.5cm]{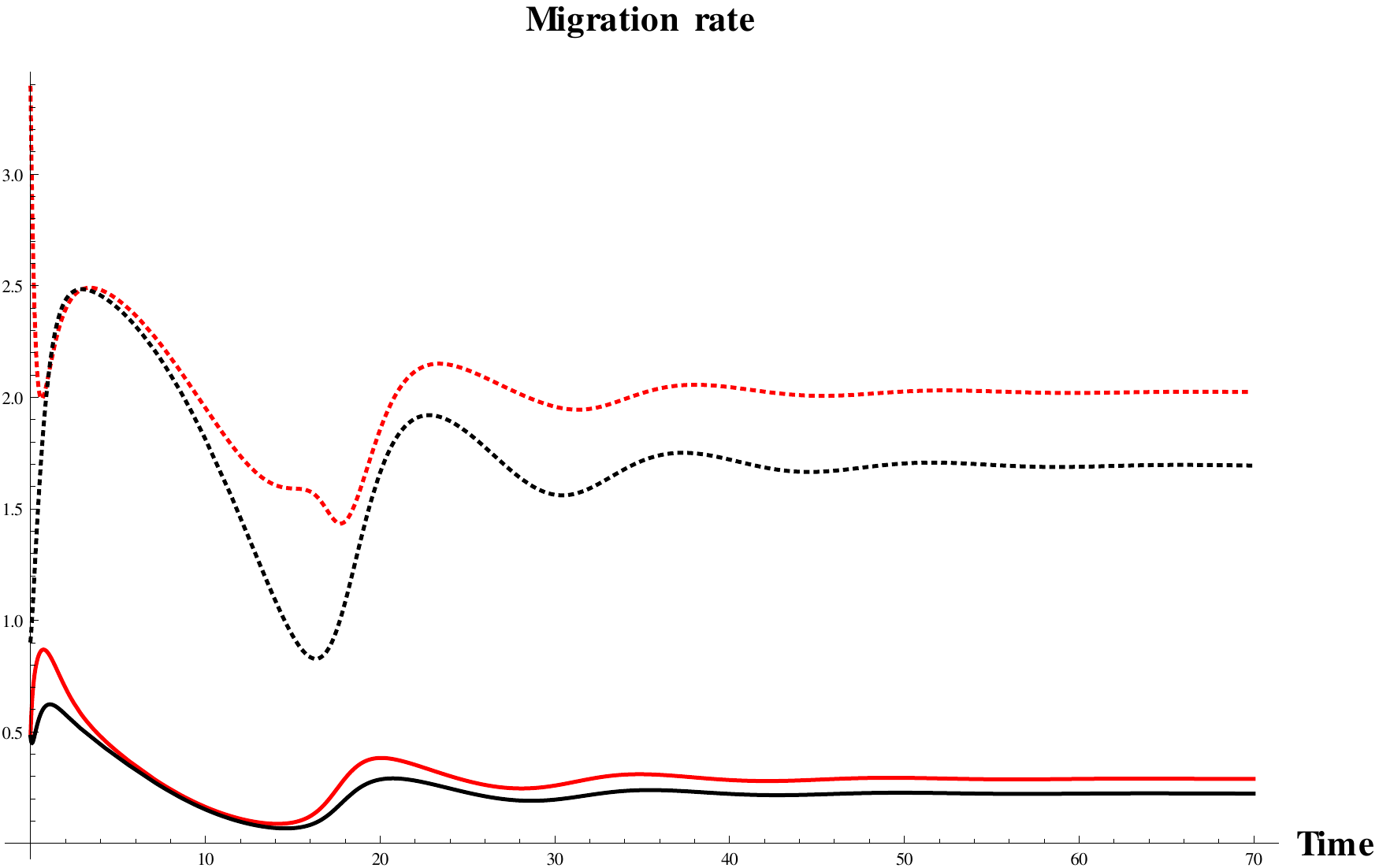}
 \caption{Number of individuals carrying each genotype and their migration rates under hypothesis (2). $\beta=1.1$ and $p=5$, with the parameters, $(b\beta-d)/c=2.4$. The initial conditions are $z_{AA,1}(0)=z_{AA,2}(0)=z_{Aa,2}(0)=0.1$, $z_{Aa,1}(0)=2$, $z_{aa,1}(0)=1$ and $z_{aa,2}(0)=2$.
     Colors: (a) the dynamics in the patch $1$ (resp. $2$) are represented in red (resp. black), the dynamics of the number of individuals with genotype $AA$ (resp $Aa$, $aa$) are represented using full (resp. dashed, dotted) lines; (b) the migration from patch $1$ to patch $2$ (resp. $2$ to $1$) 
     is drawn in red (resp. black), the migration of $aa$-individuals (resp. $AA$) is represented using dotted (resp. full) lines.
     }
\label{fig1dyn}
     \end{figure}
\begin{figure}[h]
    \centering
     \includegraphics[width=8cm,height=4cm]{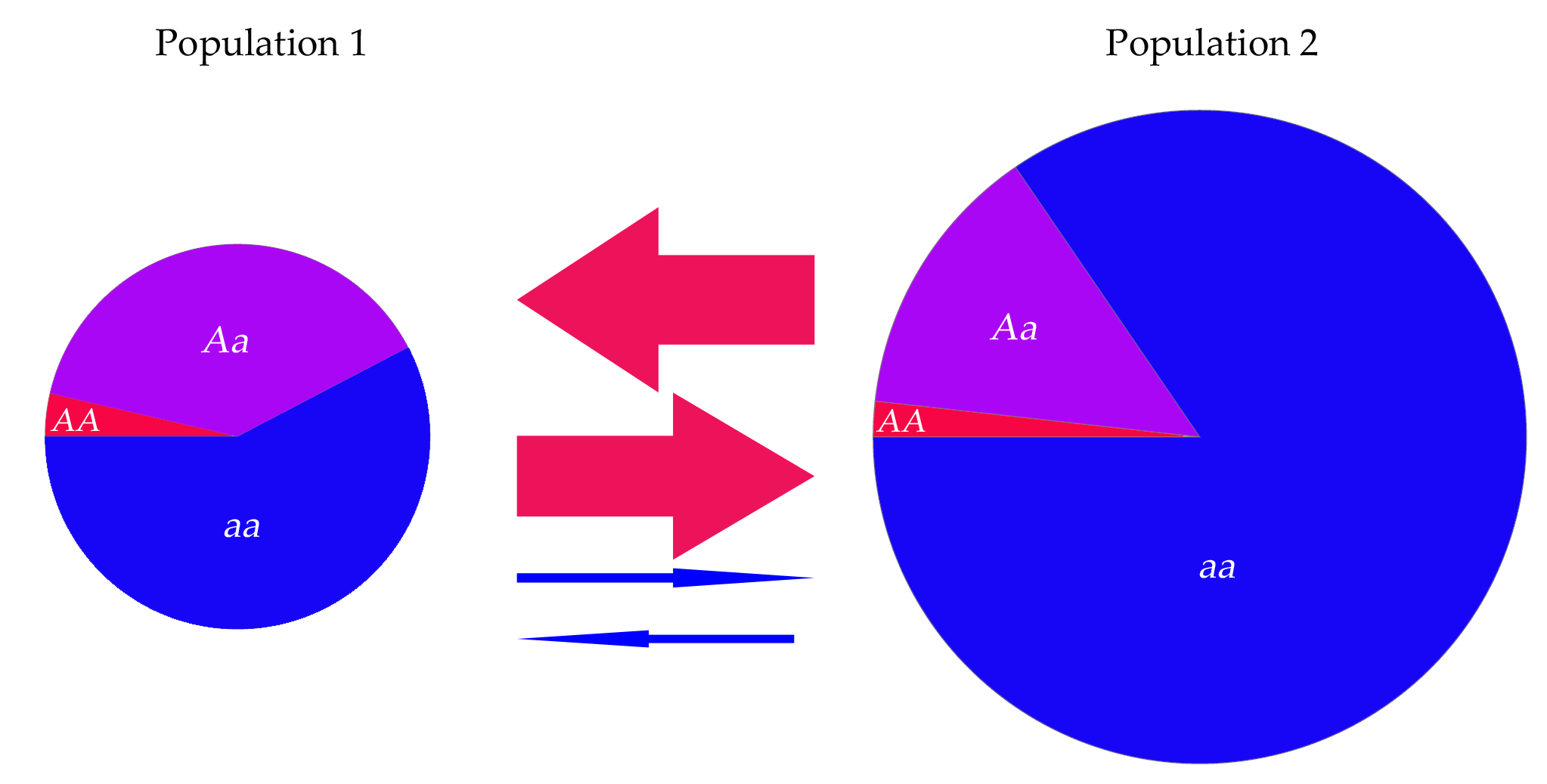}  \hfill   
     \caption{Proportions of different genotypes and migration rates between the two patches at equilibrium. $z_{AA,1}\sim 0.06, z_{Aa,1}\sim 0.64,
     z_{aa,1}\sim 0.96, z_{AA,2}\sim 0.45, z_{Aa,2}\sim 0.36, z_{aa,2}\sim 2.19$. 
     Total population in patch 1, $1.77$, total population in patch 2, $3$.
     Same parameters, initial conditions and colors than in 
     Figure \ref{fig1dyn}
    }
\label{cam1}
     \end{figure}
 This result is rather unexpected since $(\beta b-d)/c$ is the maximal equilibrium population size of an isolated population whereas $(b-d)/c$ 
 is smaller than the minimal one. Indeed, the maximal birth rate of any individual is $\beta b$, which is its birth rate when surrounded by individuals with the 
 same phenotype, whereas 
 its minimal birth rate is $b$ if surrounded by individuals with different phenotypes. In all cases, its death rate is the sum of its 
 natural death rate, $d$, 
 and of the competition death rate which is equal to the product of $c$ and of the total population size in its patch. As a consequence, 
 in an isolated patch without any migration, this leads to a maximal (resp. minimal) equilibrium equal to $(\beta b-d)/c$ (resp. larger than $(b-d)/c$).

Actually, polymorphism can be maintained through an equilibrium between growth and migration. Indeed, denoting by $N$ the population size in the less 
 populated patch, we deduce that the number of births by time unit in this patch always stays above $b*N$, whereas the number of deaths is 
  $(d+c*N)*N$. As $N$ is smaller than $(b-d)/c$, the number of deaths can not exceed $(d+(b-d))N=b*N$ and there are more births 
 than deaths in the less populated patch.
Since the equilibrium is maintained, there is a continuous excess flux of migration towards the other patch.
 Only $AA$ and $aa$ individuals migrate, which explains why the proportion of $Aa$ individuals remains relatively high.
On the contrary, as the population size of the most populated patch is larger than $(\beta b-d)/c$, there are more deaths than births in it, but the latter
constantly receives individuals of types $AA$ and $aa$ from the less populated patch, which maintains the polymorphism.

Moreover, equilibria with polymorphism are observed only when the migration rate $p$ is high with respect to the preference coefficient $\beta$ (see Fig. \ref{fig2}d), which reinforces the idea that polymorphism is maintained by a trade-off between migration and selection.

\subsection{Dominance between alleles}

We now assume a total dominance of the allele $A$, so that $Aa$ heterozygotes display the same phenotype and behaviour as $AA$ homozygotes.

\subsubsection{System without migration}

When there is no migration ($p=0$), 
the system admits two stable fixed points, $(z_{AA,1}=\zeta,z_{Aa,1}=0,z_{aa,1}=0)$ (fixation of allele $A$) and 
$(z_{AA,1}=0,z_{Aa,1}=0,z_{aa,1}=\zeta)$ (fixation of allele $a$), and one unstable fixed point with persistence of all three genotypes (see \ref{dom_sans_mig}).
In this case we were only able to give a sufficient condition on the initial number of individuals of different types for the system to converge 
to the stable fixed point characterizing the fixation of allele $A$: if $z_{AA,1}(0)\geq z_{aa,1}(0)$, the solution converges 
to the stable equilibrium $(\zeta,0,0)$.

\subsubsection{Fixed points of the system with migration} \label{fixed_pts_dom}

When there is a migration between the two patches ($p>0$), once again the dynamics is much more complex and we were unable to obtain convergence results. 
However we were able to describe some of the fixed 
points and determine their stability (see \ref{dom_mig}).
There are four fixed points with monomorphic populations in both patches, which are the same as in the case of codominance: 
fixation of $A$ in both patches $(z_{AA,1}=\zeta,z_{Aa,1}=0,z_{aa,1}=0,z_{AA,2}=\zeta,z_{Aa,2}=0,z_{aa,2}=0)$, 
fixation of $a$ in both patches $(0,0,\zeta,0,0,\zeta)$, or fixation of different alleles in the two patches, 
$(\zeta, 0,0,0,0,\zeta)$ and $(0,0,\zeta,\zeta,0,0)$.
The first two fixed points are stable for all the parameters values. 
The two last fixed points admit five negative eigenvalues and one null eigenvalue and we were not able to conclude on their stability. 
Since the recessive allele $a$ has no influence on the phenotypes of $Aa$ heterozygotes, it behaves neutrally with respect to mating and migration behaviour 
when occurring in heterozygotes. This neutral behaviour may substantially complexify the two-populations dynamics observed here.

Numerical simulations were then performed to study the influence of preference and migration parameters on the equilibria reached.

\subsubsection{Fixation of the dominant allele A}
In many simulations assuming an asymmetrical initial state (more $a$ alleles in one patch and more $A$ alleles in the other patch), the fixation of the dominant allele $A$ was observed, when preference was strictly larger than 1 
(Fig. \ref{fig5}a). Dominance of allele $A$ makes the phenotype $A$ (displayed by both $AA$ and $Aa$ genotypes) more frequent therefore provoking its fixation through 
assortative mating advantage. When allele $A$ was in minority at initial state however, its fixation happened in much less simulations, because the large initial 
frequency of genotypes $aa$ displaying phenotype a overrides the dominance effect (Fig. \ref{fig5}b). 

\begin{figure}[h]
    \centering
     \includegraphics[width=9cm,height=5cm]{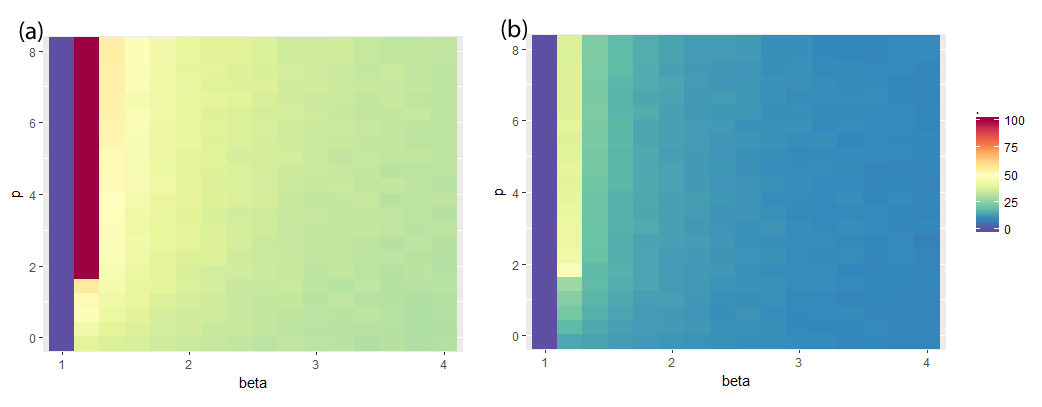}
     \caption{Conditions of fixation of the dominant allele $A$ in both
 populations. Columns differ in initial conditions. First column:
Asymmetrical frequency of allele $a$ (more frequent in population
 2); Second column: Frequency of allele $a$ greater than 0.5 in both
populations.
The color indicates the percentage of simulations where $A$ get fixed in both populations.}
   \label{fig5}
\end{figure}

Interestingly, although increasing migration promotes the fixation of the dominant allele $A$ through its homogenizing effect, increasing preference tends to 
weaken this fixation. This non-trivial effect may stem from the limitation of migration when homogeneous population emerges: individuals matching the predominant 
phenotype within a patch migrate less.  Therefore, when preference is strong and populations are initially differentiated, equilibria with fixation of 
different alleles in the two populations might be frequent, despite the frequency-dependent advantage of the phenotype carried by the dominant allele.

\subsubsection{Conditions for differentiated populations}
\label{sec:diffpopDOM}

\begin{figure}[h]
    \centering
     \includegraphics[width=9cm,height=5cm]{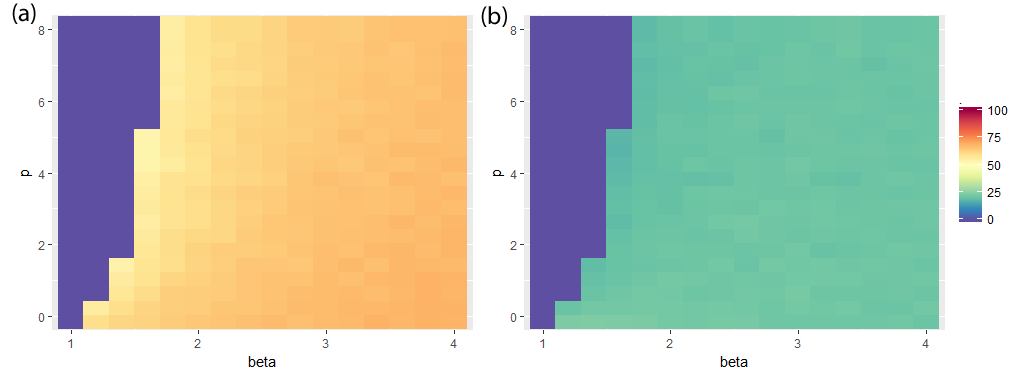}
     \caption{Conditions for the emergence of differentiated populations
(with fixation of allele $A$ in population 1 and $a$ in population 2).
 Columns differ in initial conditions. First column: Asymmetrical
frequency of allele $a$ (more frequent in population 2); Second column:
Frequency of allele $a$ greater than 0.5 in both populations.
The colors indicate the percentage of simulations where $A$ get fixed in population $1$ and $a$ in population $2$.}
   \label{fig6}
\end{figure}

As observed in Fig. \ref{fig6}a, simulations with initial differentiation mainly result in the fixation of different alleles in populations 1 and 2 when the preference 
coefficient $\beta$ increases and migration strength is not too high. 
However, it is worth noting that differentiated populations subsist for values of the migration parameter $p$ much higher than in both codominant cases 
(see Fig. \ref{fig1}).
When allele $a$ is predominant at initial state in both patches, 
simulations mostly lead to the fixation of the 
recessive allele $a$ (data not shown).
Note that around $25 \%$ of simulations lead to differentiated populations whereas in the haploid version of the model studied in 
(\cite{coron2016stochastic}) differentiated populations could emerge only when the initial state was asymmetrical.
Note also that in some parameters regions (small preference parameter $\beta$ or high migration strength $p$) no simulation leads to 
differentiated populations. This may indicate that equilibria $(\zeta,0,0,0,0,\zeta)$ and $(0,0,\zeta,\zeta,0,0)$ (see Section \ref{fixed_pts_dom}) 
are unstable for such parameters.

\subsubsection{Polymorphic equilibria}

 \begin{figure}[h]
    \centering
     \includegraphics[width=8cm,height=3.5cm]{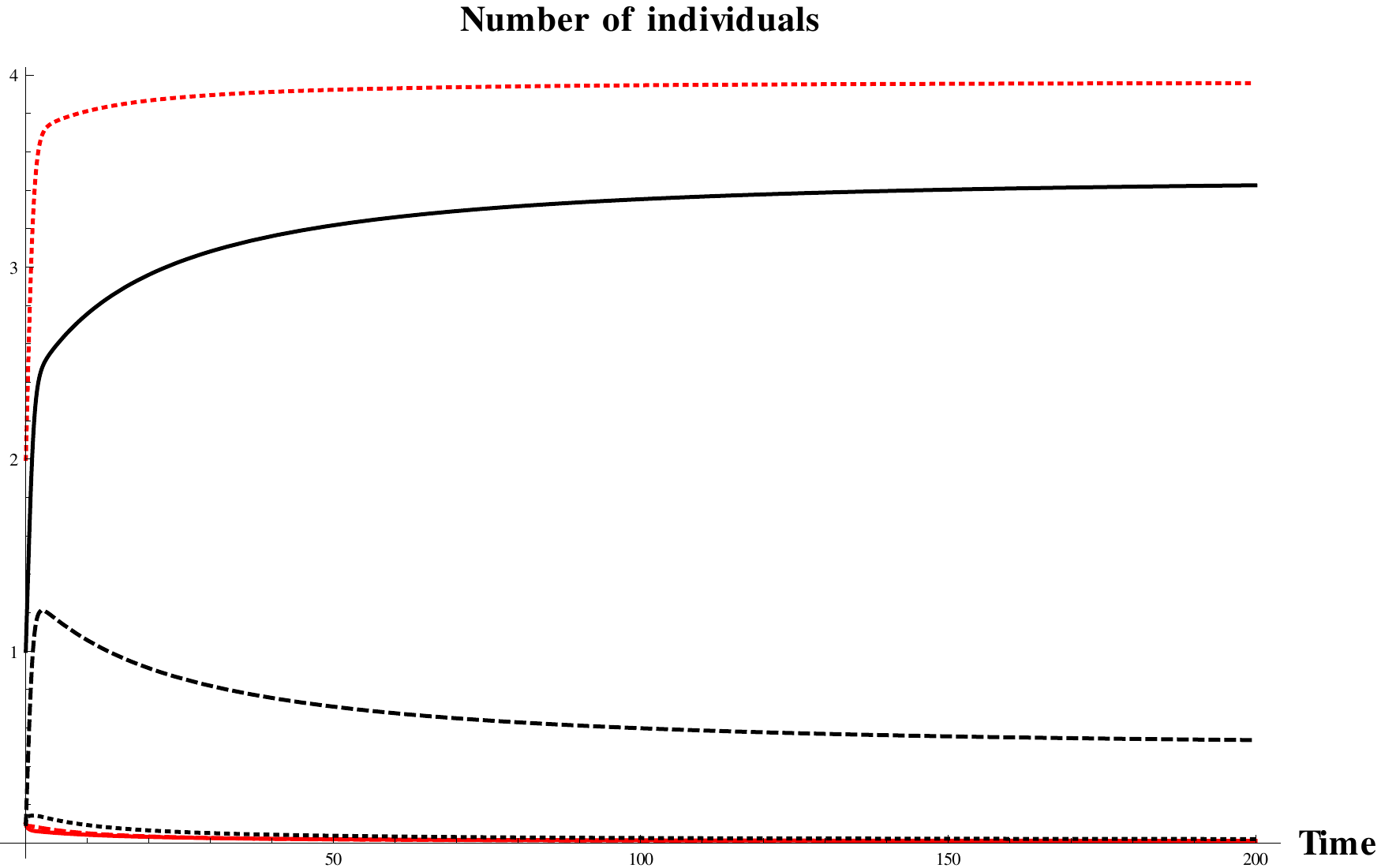}  \hfill \\ \hspace{.5cm}  \includegraphics[width=8cm,height=3.5cm]{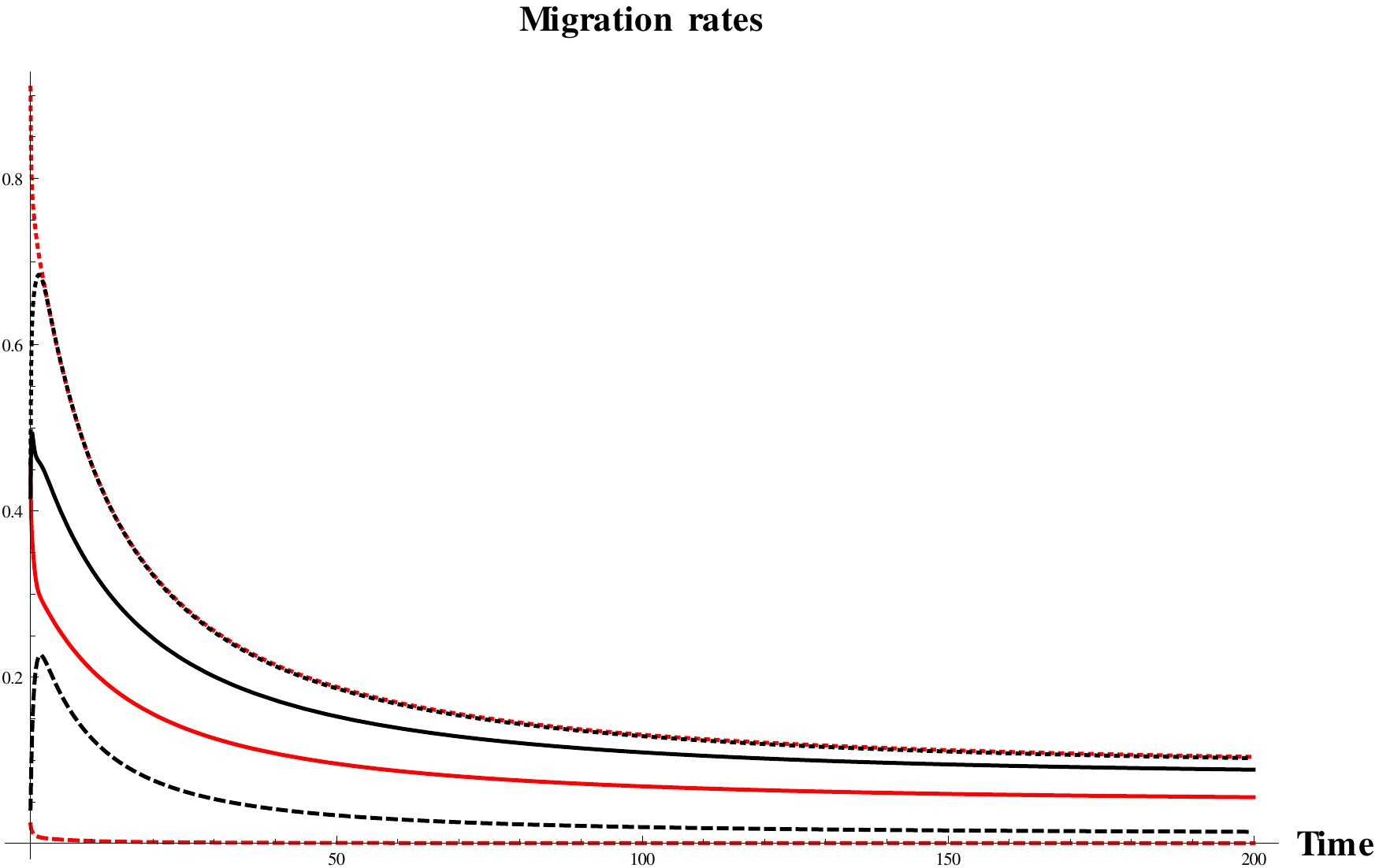}
     \caption{Population sizes and migrations under dominance hypothesis.
     $\beta=1.5$ and $p=5$. The initial conditions are $z_{aa,1}(0)=0.5$ and $0.1$ for the initial number of individuals 
     with the other genotypes.
     Colors: (a) the dynamics in the patch $1$ (resp. $2$) are represented in red (resp. black), the dynamics of the number of individuals with genotype 
     $AA$ (resp $Aa$, $aa$) are represented using full (resp. dashed, dotted) lines; (b) the migration from patch $1$ to patch $2$ (resp. $2$ to $1$) is 
     drawn in red (resp. black), the migration of $aa$-individuals (resp. $AA$) is represented using dotted (resp. full) lines.
     }
\label{fig2dyn}
     \end{figure}
     
In some cases, even if initial conditions are asymmetrical (more $a$ alleles in one patch and more $A$ alleles in the other patch), genetic polymorphism can
persist in one or both populations. However, the populations are almost phenotypically monomorphic in both patches: almost only $aa$ individuals in one patch, 
and almost only $AA$ and $Aa$ individuals in the other one (see Fig. \ref{fig2dyn} and \ref{cam2}). 
As a consequence, individuals reproduce at their maximal birth rate $\beta b$ and do not migrate, and the population sizes in both patches 
are close to their carrying capacity $(b \beta-d)/c$.

 \begin{figure}[h]
    \centering
     \includegraphics[width=8cm,height=3.5cm]{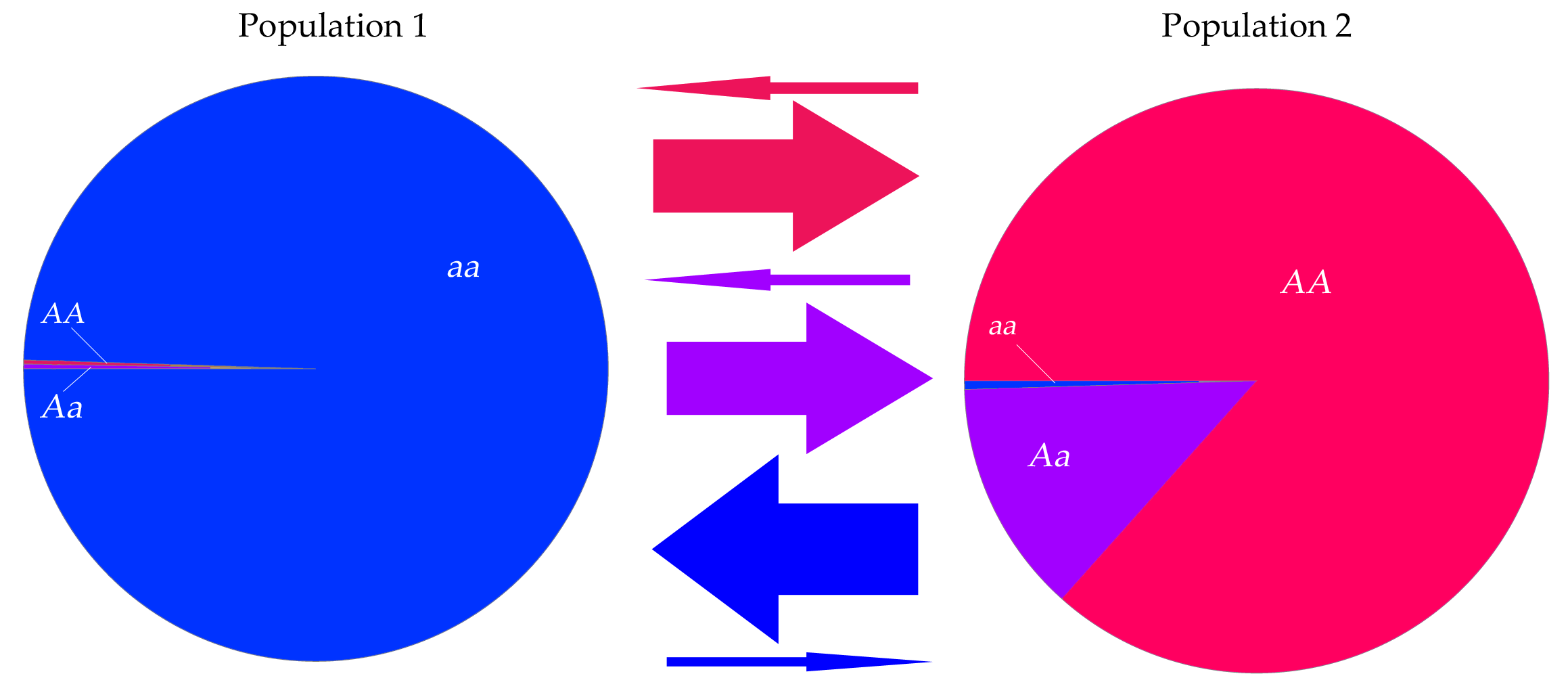}  \hfill   
     \caption{Proportions of different genotypes in the two patches at equilibrium. $z_{AA,1}\sim 0.01, z_{Aa,1}\sim 0.009,
     z_{aa,1}\sim 3.96, z_{AA,2}\sim 3.45, z_{Aa,2}\sim 0.51, z_{aa,2}\sim 0.02$. 
     Total population sizes in patch 1 and 2 very close to $(b\beta-d)/c=4$.
     Same parameters, initial conditions and colors than in 
     Figure \ref{fig2dyn}
     }
\label{cam2}
     \end{figure}

\subsection{Constant migration}

In order to understand better the role of frequency dependent migration, 
we explored the behaviours of the different genotypes when the migration rate is constant, equal to $p$ and identical for all individuals.

In this case, there are three 
possible types of equilibria $z$ (see \ref{studyconstantpim} for the proof):
\begin{itemize}
 \item[$(1)$] Either $z=0$
 \item[$(2)$] Or all the coordinates of $z$ are positive (the three genotypes are present in the two patches)
 \item[$(3)$] Or there is only one type of individuals, $AA$ or $aa$, and the population size is the same in the 
 two patches, 
\begin{equation*} \zeta:= \frac{\beta b-d}{c}. \end{equation*} 
\end{itemize}

Although a strict division of the population into two monomorphic subpopulations with different allele ($A$ in a patch and $a$ in the other one)
is not possible, there are cases where $AA$ is in large majority (more than $90\%$) in one patch, and $aa$ 
 in the other patch. This behavior can be observed under both codominant hypotheses (simulations not shown) as well as under the dominant hypothesis, as presented in Fig. \ref{figmigbasale}. However, note that the preference parameter ($\beta$) has to be very large and the migration parameter ($p$) very small. 
The ratio between the preference parameter and the migration parameter has to be much larger than in Section~\ref{sec:diffpopDOM} with our initial model.

Actually, this condition on the ratio between the preference parameter and the migration rate is necessary to observe equilibria which are not of type $(3)$ described above. In Fig.~\ref{eqnonAA_aa}, we show results assuming strict dominance between alleles.

\begin{figure}[h]
\centering
\includegraphics[width=8cm,height=3.5cm]{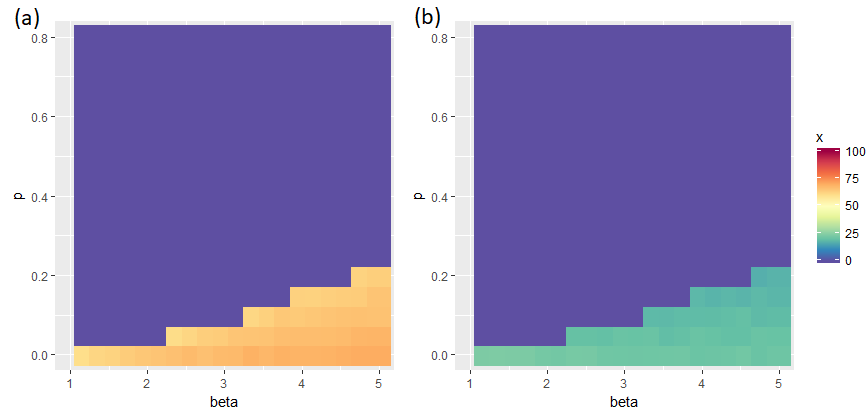}  \hfill   
     \caption{\label{figmigbasale} Conditions for the emergence of differentiated populations.
 Columns differ in initial conditions. First column: Asymmetrical
frequency of allele $a$ (more frequent in population 2); Second column:
Frequency of allele $a$ greater than 0.5 in both populations.
The colors indicate the percentage of simulations where the patch $1$ is filled with more than $90\%$ of allele $A$ in population 1 and the patch $2$ with more 
than $90\%$ of allele $a$. Note that the scales for $\beta$ and $p$ are really different from these of Fig.~\ref{fig6}}
\end{figure}

\begin{figure}[h]
\centering
\includegraphics[width=8cm,height=3.5cm]{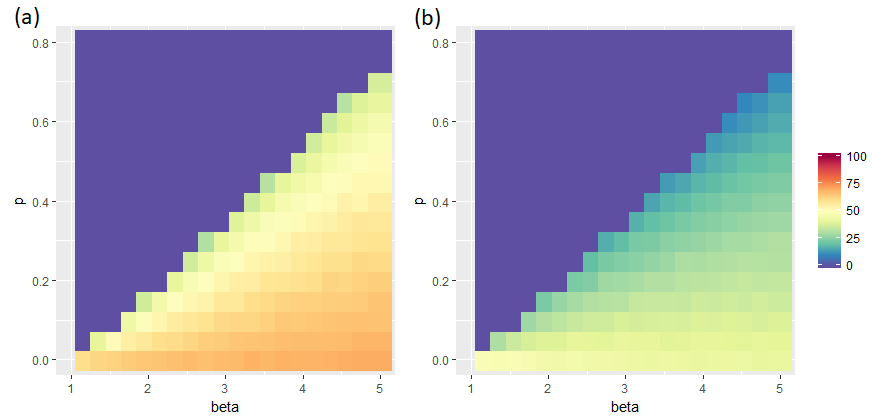}  \hfill   
     \caption{\label{eqnonAA_aa} Conditions for equilibria with all coordinates positive.
 Columns differ in initial conditions. First column: Asymmetrical
frequency of allele $a$ (more frequent in population 2); Second column:
Frequency of allele $a$ greater than 0.5 in both populations.
The colors indicate the percentage of simulations with equilibria which are not of type $(3)$.}
\end{figure}

Hence, the particular behaviours of these models with basal migration are observed only for large values of $\beta$ ($\beta$ larger than $3$ for $p=0.5$ 
for example), which contrasts substantially with our initial models. The particular migration behaviour assumed in our initial model, which depends on the 
local composition of potential mate, indeed substantially promotes population differentiation.

\section{Discussion}

\subsection{Effect of ploidy on population differentiation}
Altogether, these results show that simulations of diploid organisms assuming codominance between alleles depart from equilibria observed in the haploid model. 
Notably, migration limits the probability of population differentiation in the diploid 
model which was not observed in haploids. Populations differentiation tightly depends on initial conditions, because alleles with high frequencies are 
strongly advantaged by assortative preferences, favouring fixation of the predominant phenotype within populations. 
Because migration depends on the number of individuals with different phenotypes, population differentiation is quickly achieved in the haploid model because of a 
rapid decrease of migration as soon as differentiation starts, until a complete isolation of the two populations, which fixed different alleles. In the diploid model
however, when the two alleles are codominant, the presence of heterozygotes with intermediate phenotypes promotes migration even when populations are initially 
differentiated, and frequently leads to the fixation of a single allele throughout both connected populations. 
The discrepancy observed here between haploid and diploid assumptions highlights the need to consider the effect of ploidy in spatially-structured models of trait 
evolution, because the presence of an intermediate phenotype can interfere in the differentiation process.

\subsection{Effect of intermediate phenotypes' behaviour on polymorphism}

Depending on the assumption regarding codominance, the resulting equilibrium slightly differed.
The absence of preference of heterozygotes for their own phenotypes (hypothesis (COD2)) can promote polymorphism of the trait under sexual selection, notably when populations 
are initially uneven. These heterozygotes do not migrate and reproduce equally with any genotypes. When they initially occur in significant proportion within 
populations, they promote migration of homozygous genotypes and limits fixation of the initially predominant allele. This contrasts with the co-dominance drawn from hypothesis (COD1), where heterozygotes can migrate and are half less preferred by homozygotes. In this case, fixation of a single allele throughout both populations is  always achieved when the same allele is predominant in both populations at initial state.

Heterozygotes behaviour is therefore a key parameter in the dynamics of population differentiation in mating traits. However, empirical data on mating behaviour of intermediate 
phenotypes are scare. In \textit{Heliconius} butterflies, wing colour pattern is known to be an important visual cue for choosing mate. In the specific case of \textit{Heliconius heurippa}, 
which displays a red and white colour pattern, which can be obtained by crossing its 
red sister species \textit{H. melpomene} to the white sister species \textit{H. cydno}, assortative mate preferences have been demonstrated
(\cite{mavarez2006speciation}). This assortative 
mating has been hypothesized to favour the emergence of the species \textit{H. heurippa}, putatively created by hybridization events between 
\textit{H. melpomene} and \textit{H. cydno} (\cite{jiggins2008hybrid}). Similarly, the golden-crowned
manakin species (\textit{Lepidothrix vilasboasi}) has recently been demonstrated to have emerged from an hybrid speciation between the snow-capped 
(\textit{Lepidothrix nattereri}) and opal-crowned (\textit{Lepidothrix iris}) manakins of the Amazon basin, leading to an intermediate phenotype 
(\cite{barrera2017hybrid}). Female choice has been suggested to play a key role in the persistence of this new phenotype, driving the emergence of a new species.

We therefore hope that our predictions on the important impact of co-dominant heterozygotes behaviour on the evolution of mating trait differentiation will motivate further empirical research on mating and migration behaviour of intermediate phenotypes.

\subsection{Evolutionary consequences of dominance for population differentiation and speciation}
The fixation of the recessive allele $a$ when initially predominant in both populations is rarer when the alternative allele is dominant as compared to 
codominant. Dominant allele spreads 
among populations more easily than codominant one because of Haldane's sieve effect (\cite{haldane1927mathematical}). Assuming strict dominance, fixation of the dominant allele throughout both populations
is thus frequent, and favoured when predominant in both populations (data not shown). 

However, the invasion of the dominant haplotype throughout both populations is limited when initial populations display uneven proportions of the two alleles, leading to either (1) population differentiation or (2) persistent polymorphism within populations. 

(1) Dominance may reinforce population differentiation, because heterozygotes display the dominant phenotype and therefore rapidly increase the number of 
individuals 
carrying this phenotype within one population, therefore causing a decrease in migration between populations. Consequently, population differentiation is more 
frequently observed when one allele is dominant as compared to co-dominant, even when assuming a high migration rate.

(2) Depending on the initial distribution of heterozygotes among populations, polymorphism can also be maintained within a population where the dominant phenotype is frequent: heterozygotes then display the 
preferred phenotype and therefore do not suffer from mate rejection and scarcely migrates. This last result is in accordance with a recent paper (\cite{schneider2016diploid}), where 
the authors explore 
the consequences of dominance at loci involved in genetic incompatibilities on the dynamics of speciation in a spatially-explicit individually centered model. 
They observed that the distance between mates necessary for a spatial mosaic of species to emerge needed to be more restricted in model assuming diploidy with 
strict dominance as compared to haploid model. This highlights how dominance  may modulate spatial differentiation and emergence of well-separated species. 
Altogether this stresses out the need for diploid models of speciation, and should stimulate empirical comparisons of speciation dynamics driven by adaptive 
or sexually-selected traits displaying contrasted dominance relationships.

\section*{Aknowledgements}

This work  was funded by the Chair "Mod\'elisation Math\'ematique et Biodiversit\'e" of VEOLIA-Ecole Polytechnique-MNHN-F.X 
to HL and CS, the Young Researcher ANR DOMEVOL (ANR-13-JSV7-0003-01) and the Emergence program of Paris city council to VL.

 \appendix

\section{Definition of the models} \label{def_model}

In these Supplementary Materials, we will 
comment the models, address the question of fitnesses, and prove the results presented in the main text.\\

Recall that we denote the number of 
individuals of type $\alpha$ in the patch $i$ at time $t$ for any $\alpha\in \{AA,Aa,aa\}$, $i\in\{1,2\}$ and $t\geq 0$ by $z_{\alpha,i}(t)$. 
Moreover we denote the total number of individuals in the patch $i$ at time $t$ by 
$$N_i(t):=z_{AA,i}(t)+z_{Aa,i}(t)+z_{aa,i}(t).$$
Finally, we recall that the parameter 
$$
\zeta:=\frac{b\beta-d}{c}
$$
is the equilibrium size of a monomorphic $AA$ or $aa$ population. It will be a characteristic quantity in many equilibria for the three dynamical systems.

Recall that the models have been described in Section \ref{section_model}.
In the codominant case, we consider two different models, depending on the heterozygotes behaviour: 
\begin{enumerate}
 \item[(1)] Either the preference expressed by an individual is an 'average preference' of 
its alleles ($\beta$ for pairs $(AA,AA)$, $(aa,aa)$ and $(Aa,Aa)$, $(\beta+1)/2$ for pairs $(AA,Aa)$ and $(aa,Aa)$, and $1$ for pairs $(AA,aa)$)
and the migration rate follows the same rule,
 \item[(2)] Or heterozygotes express no preference and thus do not migrate (in this case only the homozygotes express a 
 preference towards individuals of the same genotype).
The preference parameter is thus $\beta$ for pairs $(AA,AA)$ and $(aa,aa)$, and $1$ for the other pairs.
 \end{enumerate}
We get for the first model (COD1):
\begin{equation}
\label{systcompletcodom}
\left\{
 \begin{aligned}
 \dot z_{AA,i}&=\frac{b }{N_i}\left( \beta z_{AA,i}^2+\frac{\beta+1}{2}z_{AA,i}z_{Aa,i}+\frac{\beta}{4} z_{Aa,i}^2 \right) -(d+cN_i)z_{AA,i}\\
 & \qquad   \qquad  -p\frac{z_{aa,i}+z_{Aa,i}/2}{N_i}z_{AA,i}+p \frac{z_{aa,j}+z_{Aa,j}/2}{N_j}z_{AA,j}\\
 \dot z_{Aa,i}&=\frac{b}{N_i}\left( \frac{\beta}{2} z_{Aa,i}^2+\frac{\beta+1}{2}z_{Aa,i}(z_{AA,i}+z_{aa,i})+2 z_{AA,i}z_{aa,i} \right)\\
&-(d+cN_i)z_{Aa,i} -p\frac{z_{aa,i}+z_{AA,i}}{2N_i}z_{Aa,i}+p \frac{z_{aa,j}+z_{AA,j}}{2N_j}z_{Aa,j}\\
\dot z_{aa,i}&=\frac{b }{N_i}\left( \beta z_{aa,i}^2+\frac{\beta+1}{2}z_{aa,i}z_{Aa,i}+\frac{\beta}{4} z_{Aa,i}^2 \right)-(d+cN_i)z_{aa,i}\\
&\qquad   \qquad  -p\frac{z_{AA,i}+z_{Aa,i}/2}{N_i}z_{aa,i}+p \frac{z_{AA,j}+z_{Aa,j}/2}{N_j}z_{aa,j}
\end{aligned}
\right. .
\end{equation}
The second codominant model writes
\begin{equation}
\label{systcompletcodom2}
\left\{
 \begin{aligned}
 \dot{z}_{AA,i}&= \frac{b}{N_i}\left(\beta z_{AA,i}^2+\frac{z_{Aa,i}^2}{4}+z_{Aa,i} z_{AA,i} \right)- (d+c N_i)z_{AA,i}\\
& \qquad \qquad \qquad -p \frac{ 
z_{aa,i}+z_{Aa,i}}{N_i} z_{AA,i}+p \frac{z_{aa,j}+z_{Aa,j}}{N_j} z_{AA,j}\\
\dot{z}_{Aa,i}&=\frac{b}{N_i}\left((z_{AA,i} + z_{aa,i})z_{Aa,i} + \frac{z_{Aa,i}^2}{2} + 2 z_{AA,i}z_{aa,i}\right) - (d + cN_i)z_{Aa,i}\\
\dot{z}_{aa,i}&= \frac{b}{N_i}\left(\beta z_{aa,i}^2+\frac{z_{Aa,i}^2}{4}+z_{Aa,i} z_{aa,i} \right)-(d+c N_i)z_{aa,i} \\
& \qquad \qquad \qquad-p \frac{ 
z_{AA,i}+z_{Aa,i}}{N_i} z_{aa,i}+p \frac{z_{AA,j}+z_{Aa,j}}{N_j} z_{aa,j}.
\end{aligned}
\right.
\end{equation}

In the dominant case, only individuals $aa$ express the phenotype $a$. Other individuals are of phenotype $A$.
The population dynamics in this case writes:
\begin{equation}
\label{systcompletdom}
\left\{
 \begin{aligned}
 \dot z_{AA,i}&=\frac{b\beta }{N_i}\left( z_{AA,i}+\frac{1}{2}z_{Aa,i}\right)^2 -(d+cN_i)z_{AA,i} -p\frac{z_{aa,i}}{N_i}z_{AA,i}\\
 &\qquad +p \frac{z_{aa,j}}{N_j}z_{AA,j}\\
 \dot z_{Aa,i}&=\frac{b}{N_i}( \beta z_{Aa,i}+2z_{aa,i})\left(z_{AA,i}+\frac{z_{Aa,i}}{2} \right)-(d+cN_i)z_{Aa,i}\\
&\qquad  \qquad \qquad  \qquad  \qquad -p\frac{z_{aa,i}}{N_i}z_{Aa,i}+p \frac{z_{aa,j}}{N_j}z_{Aa,j}\\
\dot z_{aa,i}&=\frac{b }{N_i}\left( \beta z_{aa,i}^2+z_{aa,i}z_{Aa,i}+\frac{\beta}{4} z_{Aa,i}^2 \right)-(d+cN_i)z_{aa,i}\\
& \qquad  \qquad -p\frac{z_{AA,i}+z_{Aa,i}}{N_i}z_{aa,i}+p \frac{z_{AA,j}+z_{Aa,j}}{N_j}z_{aa,j},
\end{aligned}
\right.
\end{equation}

Notice that the dynamical systems governing the population dynamics in the co-dominant and the dominant cases 
can be obtained as large population limits of stochastic individual based models (see \cite{coron2016stochastic}).

Precisely, let $K$ be a large parameter that gives the order size of the population. The microscopic population is represented by the process 
$(Z^K_{\alpha,i}(t))_{\alpha \in\{AA,Aa,aa\},i\in \{1,2\}}$ where $Z^K_{\alpha,i}(t)$ gives the size, divided by $K$, of the $\alpha$-population in patch $i$ at time $t$. The reproduction and migration mechanisms are described in Section~\ref{section_model} and the death rate of any individual in 
patch $i$ is given by
$$
d+c\left(Z^K_{AA,i}(t)+Z^K_{Aa,i}(t)+Z^K_{aa,i}(t)\right).
$$
Under the assumption that the sequence of initial conditions 
$\big((Z^K_{\alpha,i}(0))_{\alpha \in\{AA,Aa,aa\},i\in \{1,2\}}\big)_{K\in \N}$ converges (in probability) when $K$ goes to infinity, the sequence of stochastic 
functions $\big((Z^K_{\alpha,i}(t))_{\alpha \in\{AA,Aa,aa\},i\in \{1,2\}}, t\in [0,T]\big)_{K\in \N}$ also converges (in probability for the uniform convergence) to 
the trajectory of the deterministic models presented above. 
In Fig.~\ref{aleaVSdetcodom1} and~\ref{aleaVSdetcodom2}, we present a realisation of the trajectory of $Z^K_{aa,1}$ for different values of $K$ and under the 
parameters values presented in Fig.~\ref{fig0dyn} and \ref{fig1dyn}.
Notice that $K=10000$ is enough for the limiting deterministic model to be a very good approximation of the stochastic process.
\begin{figure}[h]
\begin{center}
\includegraphics[width=8cm,height=3.5cm]{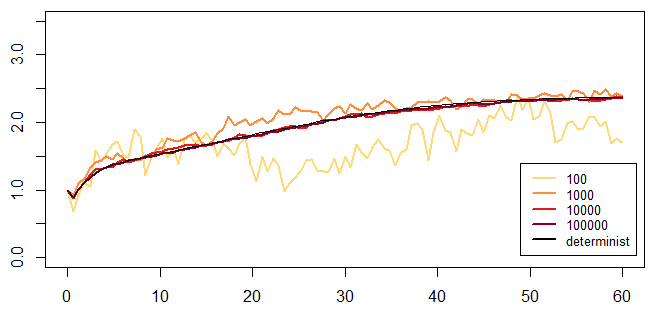}
\end{center}
\caption{\label{aleaVSdetcodom1} Comparison of the probabilistic and the deterministic COD1 models. 
Colored curves: realisations of the trajectory $(Z^K_{aa,1}(t),t\geq 0)$ for five different values of $K$ given in the legend; black curve: deterministic trajectory. The parameters are the ones of Fig.~\ref{fig0dyn}. }
\end{figure}

\begin{figure}[h]
\begin{center}
\includegraphics[width=8cm,height=3.5cm]{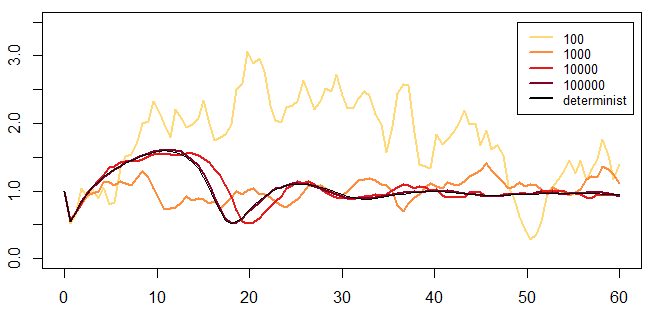}
\end{center}
\caption{\label{aleaVSdetcodom2} Comparison of the probabilistic and the deterministic COD2 models. 
Colored curves: realisations of the trajectory $(Z^K_{aa,1}(t),t\geq 0)$ for five different values of $K$ given in the legend; black curve: deterministic trajectory. The parameters are the ones of Fig.~\ref{fig1dyn}.}
\end{figure}

The mechanism of mating preference presented in (\cite{coron2016stochastic}) is similar to the ones classically used in ecology literature (see for instance 
(\cite{gavrilets1998evolution,matessi2002long,BurgerSchneider2006,Servedio2010}) and references therein).
The non usual form of the equations comes from the fact that we model a varying size population evolving in continuous time and with overlapping generations, 
whereas classical models consider discrete non overlapping generations models with an infinite population size 
(see Section 'Discussion of the model' in (\cite{coron2016stochastic}) for a detailed comparison).

As in \cite{coron2016stochastic} and in order to compare both continuous and discrete models, we give the probabilities that the individuals with genotype $\mathfrak{g}$ mate with any individual with genotype $\mathfrak{g}'$ in the deme $i$ at time $t$:
\begin{equation*}
\frac{p_\beta(\mathfrak{g},\mathfrak{g}')z_{\mathfrak{g},i}^2}{N_i\big(p_\beta(\mathfrak{g},AA)z_{AA,i}+p_\beta(\mathfrak{g},Aa)z_{Aa,i}+p_\beta(\mathfrak{g},aa)z_{aa,i}\big)}.
\end{equation*}

Finally, to ensure the survival of the population and that the total population size remains bounded, we make the following assumptions: 
$$ b>d>0 \quad \text{and} \quad c>0. $$
We recall that we are interested in the case of assortative mating, which means that 
$$ \beta \geq 1. $$

\section{On fitness proxies}

The question of defining a fitness proxy for diploid individuals with a density dependent sexual reproduction is tricky.
Fitnesses are mostly defined in the case of haploid individuals with clonal reproduction. It thus essentially consists in computing the 
exponential growth rate of a certain type of individuals (for instance mutants) in a resident population.
It can also be computed for individuals in a population with more than two types. Hence it may be necessary to take into account the current population sizes.

A first solution in our case could be to compute the relative rates at wich $AA$ and $aa$ individuals take part in an event of reproduction, as 
an $AA$ (resp. $aa$) individual necessarily transmits an allele $A$ (resp. $a$). However, such a fitness proxy would not allow to take properly into 
account the role of the heterozygote individuals, especially when one of the alleles is dominant.

Very few fitness proxies to the diploid case have been available so far (but see 
\cite{roze2005inbreeding,ravigne2006selective,parvinen2008novel,parvinen2016fitness}). 
These papers essentially deal with meta population and dispersers.
We will borrow some idea of (\cite{parvinen2008novel}) and of branching process theory 
(see \cite{athreya1972branching} for instance) to propose a fitness proxy in our case.

In (\cite{parvinen2008novel}), the authors devise a fitness proxy for a set of metapopulation 
models defined in continuous time. It corresponds to the expected number of mutant dispersers produced by a local mutant population initiated by one 
mutant disperser and is computed via the principal eigenvalue of a matrix taking into account the fact that a mutant allele may be in a homozygote or an 
heterozygote mutant.

Our case is more involved however, because the birth rates are type and density dependent.
Our strategy is thus the following: we consider an environment set by the current population sizes in a patch, construct a multitype branching process 
with this environment, and define relative fitnesses as the long term proportions of individuals. It has to be understood that such fitnesses are 
relative. We aim at finding which genotype is locally favored in a patch among $AA$, $Aa$, and $aa$, and not giving a quantitative fitness value to each 
genotype.

To construct this fitness proxy, let us first recall that an interpretation of the birth rates is that half of the time individuals reproduce as a female 
(they choose their mate according to their preference), and half of the time as a male (they are choosen by the female).
For a fixed value of population sizes in a patch $(z_{AA},z_{Aa},z_{aa})$ we may thus compute the rate at which an individual with a given genotype $\sigma$
gives birth to an individual of a given genotype $\tau$, $b_{\sigma \tau}$.

In our model, these rates have the following expressions (where $z_{\alpha \alpha'}, (\alpha,\alpha') \in \{A,a\}^2$ is the $\alpha \alpha'$-population 
size in the patch considered, $\bar{\alpha}$ is the complementary of $\alpha$ in $\{A,a\}$, and $N=z_{AA}+z_{Aa}+z_{aa}$):
$$
 b_{\alpha \alpha\to \alpha \alpha}= \frac{b}{N} \beta z_{\alpha \alpha} , \quad  b_{\alpha \alpha \to \bar{\alpha}\bar{\alpha}}= 0,
$$
$$
 b_{\alpha \alpha \to \alpha\bar{\alpha}}=  \frac{b}{N} \left(p_\beta(\alpha\alpha,\alpha \bar{\alpha}) \frac{z_{\alpha \bar{\alpha}}}{2}+
 p_\beta(\alpha \alpha,\bar{\alpha}\bar{\alpha})z_{\bar{\alpha}\bar{\alpha}} \right),
$$
$$
 b_{\alpha \bar{\alpha} \to \alpha\alpha}=  \frac{b}{N} \left(p_\beta(\alpha\bar{\alpha},\alpha \bar{\alpha}) \frac{z_{\alpha \bar{\alpha}}}{4}+
 p_\beta(\alpha \bar{\alpha},\alpha\alpha)\frac{z_{\alpha\alpha}}{2} \right),
$$
$$
 b_{\alpha \bar{\alpha} \to \bar{\alpha}\bar{\alpha}}=  \frac{b}{N} \left(p_\beta(\alpha\bar{\alpha},\alpha \bar{\alpha}) \frac{z_{\alpha \bar{\alpha}}}{4}+
 p_\beta(\alpha \bar{\alpha},\bar{\alpha}\bar{\alpha})\frac{z_{\bar{\alpha}\bar{\alpha}}}{2} \right),
$$
$$
 b_{\alpha \bar{\alpha} \to \alpha\bar{\alpha}}=  \frac{b}{N} \left(p_\beta(\alpha\bar{\alpha},\alpha \bar{\alpha}) \frac{z_{\alpha \bar{\alpha}}}{2}+
 p_\beta(\alpha \bar{\alpha},\alpha\alpha)\frac{z_{\alpha\alpha}}{2} + p_\beta(\alpha \bar{\alpha},\bar{\alpha}\bar{\alpha})\frac{z_{\bar{\alpha}\bar{\alpha}}}{2} \right).
$$
These rates can be seen as the birth rates of a continuous time multitype branching process. 
If we introduce the matrix
$$M:= \small{ \left( \begin{array}{ccc}
     b_{AA\to AA} & b_{AA\to Aa} & b_{AA\to aa} \\
 b_{Aa\to AA} & b_{Aa\to Aa} & b_{Aa\to aa} \\
  b_{aa\to AA} & b_{aa\to Aa} & b_{aa\to aa}
   \end{array}\right)},
 $$
We know that this matrix has a positive maximal eigenvalue $\lambda$, a left and a right eigenvectors $u$ and $v$ associated to $\lambda$ 
with positive coordinates such that 
$$ u.v=1 \quad \text{and} \quad u.1=1. $$
Then the total population size of the branching process grows exponentially with a Malthusian parameter $\lambda$ and the proportions of the populations
$(z_{AA},z_{Aa},z_{aa})$ converge to $u$.
These proportions could be a proxy for the fitnesses of the different genotypes in a patch at a given moment.

\section{Behaviour of the system without migration}

Let us first study the behaviour of the system when there is no migration. 
In this case the two patches have independent dynamics, and it is enough to study 
one patch.
\subsection{First codominant case} \label{cod1_sans_mig}
\label{ssec_COD1}
We will prove the following result:

\begin{lem}\label{lemmecod}
 If $p=0$, there are two stable and one unstable fixed points in the patch $i$ for $i \in \{1,2\}$:
 $$ (\zeta,0,0), \quad (0,0,\zeta) \quad \text{and} \quad  (\delta \xi,\xi, \delta \xi), $$
 where
$\delta$ is the unique positive root of the polynomial functional 
$$ P(X):=X^3+\frac{1}{2}X^2-\frac{1}{4}X-\frac{\beta}{8},$$
and
$$
\xi:=\frac{b(\delta+\beta/2)-d}{(2\delta+1)c}.
$$
Moreover, we have the following asymptotic behaviours for the dynamical system \eqref{systcompletcodom}:
\begin{itemize}
 \item If $z_{AA,i}(0)>z_{aa,i}(0)$, then the solution converges to the stable fixed point $(\zeta,0,0)$
 \item If $z_{AA,i}(0)=z_{aa,i}(0)$, then the solution converges to the unstable fixed point $(\delta \xi,\xi, \delta \xi)$.
 \item If $z_{AA,i}(0)<z_{aa,i}(0)$, then the solution converges to the stable fixed point $(0,0,\zeta)$
 \end{itemize}
 \end{lem}

\begin{proof}

To begin with, we describe the different stable fixed points in the codominant case.
We recall that the dynamics of the system in one patch is given by the system of equations:
\begin{equation}
\label{sys1patch0}
\left\{
 \begin{aligned}
& \dot z_{AA}=\frac{b }{N}\left( \beta z_{AA}^2+\frac{\beta+1}{2}z_{AA}z_{Aa}+\frac{\beta}{4} z_{Aa}^2 \right) -(d+cN)z_{AA}\\
& \dot z_{Aa}=\frac{b}{N}\left( \frac{\beta}{2} z_{Aa}^2+\frac{\beta+1}{2}z_{Aa}(z_{AA}+z_{aa})+2 z_{AA}z_{aa} \right)-(d+cN)z_{Aa}\\
& \dot z_{aa}=\frac{b }{N}\left( \beta z_{aa}^2+\frac{\beta+1}{2}z_{aa}z_{Aa}+\frac{\beta}{4} z_{Aa}^2 \right)-(d+cN)z_{aa},
 \end{aligned}
\right.
\end{equation}
where $N=z_{AA}+z_{Aa}+z_{aa}$.

The fixed points $(z_{AA},z_{Aa},z_{aa})$ are solutions to the following system of equations:
\begin{equation}
\label{sys1patch}
\left\{
 \begin{aligned}
&\frac{b }{N}\left( \beta z_{AA}^2+\frac{\beta+1}{2}z_{AA}z_{Aa}+\frac{\beta}{4} z_{Aa}^2 \right) =(d+cN)z_{AA}\\
&\frac{b}{N}\left( \frac{\beta}{2} z_{Aa}^2+\frac{\beta+1}{2}z_{Aa}(z_{AA}+z_{aa})+2 z_{AA,1}z_{aa} \right)=(d+cN)z_{Aa}\\
&\frac{b }{N}\left( \beta z_{aa}^2+\frac{\beta+1}{2}z_{aa}z_{Aa}+\frac{\beta}{4} z_{Aa}^2 \right)=(d+cN)z_{aa}.\\
 \end{aligned}
\right.
\end{equation}

\subsection*{Monomorphic equilibria}

We first check that if $z_{\alpha \alpha}=0$ for $\alpha \in \{A,a\}$, then $z_{Aa}$ is necessarily equal to $0$. Hence we obtain the two following 
monomorphic fixed points:
$$ \left( \zeta,0,0 \right) \quad \text{and} \quad \left(0,0, \zeta\right). $$
The eigenvalues for these fixed points are
$$ \left(-b\beta,-\frac{b}{2}(\beta-1),-b\beta+d \right) .$$
They are all negative under our assumptions.

If $z_{Aa}>0$ then from the first and the last equations in \eqref{sys1patch} we see that necessarily $z_{AA}>0$ and $z_{aa}>0$.
Hence we look for a fixed point with the three coordinates positive.
From the first and the last equations in \eqref{sys1patch} we get that
\begin{multline*}
(d+cN)z_{AA}z_{aa}= z_{aa}(\beta z_{AA}^2+\frac{\beta+1}{2}z_{AA}z_{Aa}+\frac{\beta}{4}z_{Aa}^2)\\
=z_{AA}(\beta z_{aa}^2+\frac{\beta+1}{2}z_{AA}z_{Aa}+\frac{\beta}{4}z_{Aa}^2),
\end{multline*}
which yields
\begin{equation}
\label{eq_prodnul}
 (z_{AA}-z_{aa})(z_{AA}z_{aa}-\frac{z_{Aa}^2}{4})=0.
\end{equation}
Hence, either $4z_{AA}z_{aa}=z_{Aa}^2$ or $z_{aa}=z_{AA}$.

\subsection*{Case $4z_{AA}z_{aa}=z_{Aa}^2$}

By expressing in two different ways $(d+cN)N$ thanks to the second and last equations in \eqref{sys1patch} and replacing $z_{AA}$ by 
$z_{Aa}^2/4z_{aa}$, we get
\begin{multline*}
 \frac{1}{z_{aa}}\left(\beta z_{aa}^2+\frac{\beta+1}{2}z_{aa}z_{Aa}+\frac{\beta}{4}z_{Aa}^2\right) \\ 
 = \frac{1}{z_{Aa}}\left( \frac{\beta}{2}z_{Aa}^2+\frac{\beta+1}{2}z_{Aa}\left(\frac{z_{Aa}^2}{4z_{aa}}+z_{aa}\right)+\frac{z_{Aa}^2}{2}\right).
\end{multline*}
This implies
\begin{equation*}
 \frac{\beta-1}{2}\left( z_{aa}+\frac{z_{Aa}^2}{4z_{aa}} \right)=0,
\end{equation*}
and contradicts the fact that all the fixed point coordinates are positive.

\subsection*{Case $z_{aa}=z_{AA}$}

The fixed point is thus of the form $(x,y,x)$, with $x>0$ and $y>0$.
Equalizing the first equation of \eqref{sys1patch} divided by $x$ and the second one divided by $y$, we find
\begin{equation*}
\begin{aligned}
\left( \beta x+\frac{\beta+1}{2}y+\frac{\beta}{4}\frac{y^2}{x}\right)=\left(\frac{\beta}{2}y+(\beta+1)x+2\frac{x^2}{y}\right)\\
\Leftrightarrow \frac{x^3}{y^3}+\frac{1}{2}\frac{x^2}{y^2}-\frac{1}{4}\frac{x}{y}-\frac{\beta}{8}=0.
\end{aligned}
\end{equation*} 
 Let us set $\delta:=\frac{x}{y}>0$.  
 The polynomial function $P(X)= X^3+\frac{1}{2}X^2-\frac{1}{4}X-\frac{\beta}{8}$, has only one positive root and this latter belongs to $]1/2,+\infty[$.
 Indeed by taking the first derivative of $P$ we can check that $P(X)$ is increasing until $x=-1/2$, then decreasing until $X=1/6$ and then 
 increasing again. As $P(-1/2)=(1-\beta)/8<0$, we conclude that there is only one positive root.
 Finally we can check that $P(1/2)<0$. 
 Then using the fact that $x=\delta y$ and the second equation of \eqref{sys1patch}, we conclude
 that there exists only one equilibrium with positive coordinates which is
$$
(\delta \xi,\xi, \delta \xi) \quad \text{where} \quad \xi:=\frac{b(\delta+\beta/2)-d}{(2\delta+1)c}, 
$$
and $\delta$ is the unique positive solution of  $P(X)=0$.

Using again that $P(\delta)=0$, we can write the eigenvalues of the Jacobian Matrix at $(\delta \xi,\xi,\delta \xi)$ as
 $$ \left( \frac{b\beta}{4\delta}(2\delta-1),
 \lambda_2,\lambda_3 \right), $$
where $\lambda_2$ and $\lambda_3$ are complex numbers. 
As $\delta \in ]1/2,+\infty[$,
this ensures that the fixed point $(\delta \xi ,\xi,\delta \xi)$ is unstable.\\
We now prove that its stable manifold is of dimension $2$ and is exactly the set 
$$\mathcal{A}^= =\{{\bf{z}}\in (\R_+^*)^3, z_{AA}=z_{aa}\}.$$
By subtracting the third equation of \eqref{sys1patch0} to the first one, we get
\begin{equation}
\label{eqzAA-zaa}
\dot z_{AA}-\dot z_{aa}= (z_{AA}- z_{aa})\left( b\beta -d-cN - b \frac{\beta-1}{2} \frac{z_{Aa}}{N} \right), 
\end{equation}
and we deduce that the set $\mathcal{A}^=$ is invariant under the flow defined by \eqref{sys1patch0}. 
Thus, let us assume that the initial condition belongs to $\mathcal{A}^=$. Hence for any $t\geq 0$, $z_{AA}(t)=z_{aa}(t)$. Moreover,
\begin{equation*}
\frac{d}{dt}\left(\frac{z_{AA}(t)}{\delta z_{Aa(t)}}-1\right)=-\frac{2b}{\delta z_{Aa}^2 N}\left(z_{AA}^3+\frac{z_{Aa}}{2}z_{AA}^2-\frac{z_{Aa}^2}{4}z_{AA}-\frac{\beta z_{Aa}^3}{8}\right).
\end{equation*}
Since $P(\delta)=0$, $\delta z_{Aa}$ is a root of the polynomial function with respect to the variable $z_{AA}$ of the right hand side. Thus
\begin{multline}
\label{eqzAAsurzAa}
\frac{d}{dt}\left(\frac{z_{AA}(t)}{\delta z_{Aa(t)}}-1\right)\\
=-\frac{2b}{ z_{Aa} N}\left(\frac{z_{AA}}{\delta z_{Aa}}-1\right)\left(z_{AA}^2+\left(\frac{1}{2}+\delta\right)z_{Aa}z_{AA}+\frac{\beta z_{Aa}^2}{8\delta}\right).
\end{multline}
The polynomial function of degree two of the r.h.s in \eqref{eqzAAsurzAa} is non-negative. This ensures that the function 
$$W_1({\bf{z}})=\ln \left(\left|\frac{z_{AA}}{\delta z_{Aa}}-1\right|\right)
$$
is a Lyapounov function for the dynamical system \eqref{sys1patch0} restricted on $\mathcal{A}^=$. Theorem 1 in \cite{lasalle1960some} implies that the 
flow ${\bf{z}}(t)$ converges to $\{(\delta\xi,\xi,\delta\xi)\}$, the largest invariant set included in $\{z_{AA}=\delta z_{Aa}\}\cap \mathcal{A}^=$. 
This implies that $\mathcal{A}^=$ is included in the stable manifold of $(\delta\xi,\xi,\delta\xi)$.\\
 
 Let us now deal with the solution outside $\mathcal{A}^=$. First note that by \eqref{eqzAA-zaa}, the two sets
$$
\mathcal{A}^> = \{{\bf{z}}\in (\R_+^*)^3, z_{AA}>z_{aa}\} \quad \text{and} \quad \mathcal{A}^< = \{{\bf{z}}\in (\R_+^*)^3, z_{AA}<z_{aa}\}
$$
are two invariant sets under the dynamical system \eqref{sys1patch0}. Since the system \eqref{sys1patch0} is symetric with respect to $z_{AA}$ and $z_{aa}$, 
we only have to deal with one of the previous sets, the dynamics in the other one being symmetric. In what follows, we study the dynamics in the set $\mathcal{A}^>$.\\
On the set $\mathcal{A}^>$, the Lyapunov function
$$
W_2({\bf{z}})=\ln \left( \frac{z_{AA}+z_{aa}+\beta z_{Aa}}{z_{AA}-z_{aa}}\right)
$$
is well defined and is clearly non-negative. Moreover, its derivative is
$$
\frac{d}{dt}W_2({\bf{z}}(t))=-\frac{b(\beta-1)}{2N(z_{AA}+z_{aa}+\beta z_{Aa})}[\beta z_{Aa}(z_{AA}+z_{aa})]\leq 0,
$$
and this derivative is equal to $0$ on the set $\mathcal{O}^>=\{{\bf{z}}\in \mathcal{A}^>, z_{Aa}=0\}$. The largest invariant set with respect to \eqref{sys1patch0} including in $\mathcal{O}^>$ is $\{{\bf{z}}\in \mathcal{A}^>, z_{Aa}=z_{aa}=0\}$ and it is obvious that any trajectory starting from this invariant set converges to $(\zeta,0,0)$. Using Theorem 1 in \cite{lasalle1960some}, it is sufficient to conclude that any trajectory starting from $\mathcal{A}^>$ converges to $(\zeta,0,0)$.
\end{proof}

\subsection{Second codominant case} \label{cod2_sans_mig}
 Using the same ideas and the same proofs as in the previous section \ref{ssec_COD1},
we prove the following lemma for the codominant model when heterozygotes are not preferred and do not migrate:
\begin{lem}
 If $p=0$, there are two stable and one unstable fixed points in the patch $i$ for $i \in \{1,2\}$:
 $$ (\zeta,0,0), \quad (0,0,\zeta) \quad \text{and} \quad  (\delta' \xi',\xi', \delta' \xi'), $$
 where
$\delta'$ is the unique positive root of the polynomial functional 
$$ Q(X):=X^3+\frac{2-\beta}{2}X^2-\frac{1}{4}X-\frac{1}{8},$$
and
$$
\xi':=\frac{b(\delta'+1/2)-d}{(2\delta'+1)c}.
$$
Moreover, we have the following asymptotic behaviours for the dynamical system \eqref{systcompletcodom}:
\begin{itemize}
 \item If $z_{AA,i}(0)>z_{aa,i}(0)$, then the solution converges to the stable fixed point $(\zeta,0,0)$
 \item If $z_{AA,i}(0)=z_{aa,i}(0)$, then the solution converges to the unstable fixed point $(\delta' \xi',\xi', \delta' \xi')$.
 \item If $z_{AA,i}(0)<z_{aa,i}(0)$, then the solution converges to the stable fixed point $(0,0,\zeta)$
 \end{itemize}
\end{lem}
The proof is really similar except that we use the Lyapunov function
$$
W_3({\bf{z}})=\ln \left( \frac{z_{AA}+z_{aa}+z_{Aa}}{z_{AA}-z_{aa}}\right)
$$
to end it.

\subsection{Dominant case} \label{dom_sans_mig}

We recall that the dynamics of the system in one patch is given by the system of equations:

\begin{equation}\label{sys1patch0dom}
 \left\{
\begin{aligned}
& \dot z_{AA}=\frac{b\beta }{N}\left( z_{AA}+\frac{1}{2}z_{Aa}\right)^2 -(d+cN)z_{AA}\\
& \dot z_{Aa}=\frac{b}{N}( \beta z_{Aa}+2z_{aa})\left(z_{AA}+\frac{z_{Aa}}{2} \right)-(d+cN)z_{Aa}\\
& \dot z_{aa}=\frac{b }{N}\left( \beta z_{aa}^2+z_{aa}z_{Aa}+\frac{\beta}{4} z_{Aa}^2 \right)-(d+cN)z_{aa},
\end{aligned}
\right.
\end{equation}
where $N=z_{AA}+z_{Aa}+z_{aa}$.

In this section, we will prove the following result:
\begin{lem}
 The dynamical system \eqref{sys1patch0dom} admits exactly three fixed points:
 \begin{itemize}
  \item Two stable fixed points 
  $$ (\zeta,0,0) \quad \text{and} \quad (0,0,\zeta). $$
  \item One unstable fixed point
 $$ \frac{b(\beta+1)-2d}{4c} \left( \frac{\sqrt{\beta+1}-1 }{\sqrt{\beta+1}+1}, \frac{2}{\sqrt{\beta+1}+1},
1\right). $$
 \end{itemize}
\end{lem}

\begin{proof}
The fixed points $(z_{AA},z_{Aa},z_{aa})$ are solutions to the following system of equations:

\begin{equation} \label{sys1patchdom}
 \left\{
\begin{aligned}
 &\frac{b\beta }{N}\left( z_{AA}+\frac{1}{2}z_{Aa}\right)^2 =(d+cN)z_{AA}\\
& \frac{b}{N}( \beta z_{Aa}+2z_{aa})\left(z_{AA}+\frac{z_{Aa}}{2} \right)=(d+cN)z_{Aa}\\
& \frac{b }{N}\left( \beta z_{aa}^2+z_{aa}z_{Aa}+\frac{\beta}{4} z_{Aa}^2 \right)=(d+cN)z_{aa}
\end{aligned}
\right.
\end{equation}

Again we can show that if $z_{\alpha \alpha}=0$ for $\alpha \in \{A,a\}$, then necessarily $z_{Aa}=0$, and we find the two following monomorphic equilibria
 $$ \left( \zeta , 0, 0\right) \text{ and } \left(0,0, \zeta \right)$$
with respective eigenvalues
 $$ (0,-b \beta,d-b \beta) \text{ and } (-b \beta,b-b \beta,d-b \beta). $$
Hence the equilibrium $(0,0,\zeta)$ is stable. To find the stability of the equilibrium $(\zeta,0,0)$, we will use a Lyapunov function.

First notice that the derivative of the total population size satisfies:
$$ \frac{d}{dt}N= (b\beta-d-cN)N- 2b(\beta-1)\frac{(z_{AA}+z_{Aa})z_{aa}}{N} ,$$
and the derivative of the difference $z_{AA}-z_{aa}$ satisfies:
$$ \frac{d}{dt}(z_{AA}-z_{aa})= (z_{AA}-z_{aa})(b\beta-d-cN)+b(\beta-1)\frac{z_{Aa}z_{aa}}{N}. $$
In particular, the set $\mathcal{A}:= \{z_{AA}\geq z_{aa}\}$ is stable. 
Indeed, if we consider a point where $z_{AA}=z_{aa}$
then the previous derivative is positive and we stay in the set $\mathcal{A}$. 
Moreover in this set 
$$ \frac{d}{dt}\ln \left(\frac{z_{AA}-z_{aa}}{N} \right)= \frac{b(\beta-1)}{N}\left( \frac{z_{Aa}z_{aa}}{z_{AA}-z_{aa}}+ 2 \frac{(z_{AA}+z_{Aa})z_{aa}}{N} \right)
\geq 0 .$$
We deduce that $- \ln((z_{AA}-z_{aa})/N)$
is a Lyapunov function 
which cancels out of the set $\{z_{aa}=0\} \cup \{z_{AA}=z_{Aa}=0\}$. 
Beside the only fixed point in this set is $(\zeta,0,0)$. 
Applying Theorem 1 in \cite{lasalle1960some}, we deduce that any solution to \eqref{sys1patch0dom} starting from $\mathcal{A}$ converges to 
the fixed point $(\zeta,0,0)$.

 Let us now find and study the positive equilibrium. By subtracting the second equality to the first one in \eqref{sys1patchdom}, we get
 \begin{equation*}
 \frac{\beta}{z_{AA}}\left( z_{AA}+\frac{1}{2} z_{Aa} \right)=\frac{1}{z_{Aa}} \left( \beta z_{Aa}+2z_{aa} \right),
\end{equation*}
which gives the equality
\begin{equation}
 \label{eq_egalite1}\beta z_{Aa}^2=4 z_{AA} z_{aa}.
\end{equation}
Besides,
\begin{equation*}
 \frac{b }{z_{aa}}\left( \beta z_{aa}^2+z_{aa}z_{Aa}+\frac{\beta}{4} z_{Aa}^2 \right)=\frac{b \beta}{z_{aa}}\left(z_{aa}+\frac{1}{2}z_{Aa}\right)^2-b(\beta-1)z_{Aa}.
\end{equation*}
Hence, by subtracting the third to the first inequality in \eqref{sys1patchdom}, we find
\begin{equation*}
 \begin{aligned}
  \frac{\beta}{z_{AA}z_{aa}}\left[ z_{aa}\left( z_{AA}+\frac{1}{2}z_{Aa}\right)^2-z_{AA}\left(z_{aa}+\frac{1}{2}z_{Aa}\right)^2 \right]+(\beta-1)z_{Aa}&=0\\
\Leftrightarrow \beta (z_{aa}-z_{AA})\left(\frac{1}{4}z_{Aa}^2-z_{aa}z_{AA}\right)+(\beta-1)z_{Aa}z_{AA}z_{aa}&=0.
 \end{aligned}
\end{equation*}
Finally, by using \eqref{eq_egalite1}, we deduce
\begin{equation*}
 (\beta-1)\frac{z_{Aa}^2}{4} (z_{AA}-z_{aa}+z_{Aa})=0,
\end{equation*}
which leads to
\begin{equation}
 \label{eq_egalite2}
z_{AA}+z_{Aa}=z_{aa}.
\end{equation}
From \eqref{eq_egalite1} and \eqref{eq_egalite2}, we get
$$ z_{AA}= \frac{z_{Aa}}{2}\left( \sqrt{\beta+1}- 1 \right) \quad \text{and} \quad  z_{aa}= \frac{z_{Aa}}{2}\left( \sqrt{\beta+1}+ 1 \right) $$
If we inject these inequalities in the derivative of $z_{Aa}$, we obtain 
$$ 0=\left( b \frac{\beta+1}{2}-d-cz_{Aa}(\sqrt{\beta+1}+1) \right)z_{Aa}. $$
Hence
$$ z_{Aa}=\frac{b(\beta+1)/2-d}{c(\sqrt{\beta+1}+1)}. $$
We deduce that the positive equilibrium we are looking for, if it exists, has necessarily the following coordinates
 $$ \frac{b(\beta+1)-2d}{4c} \left( \frac{\sqrt{\beta+1}-1 }{\sqrt{\beta+1}+1}, \frac{2}{\sqrt{\beta+1}+1},
1\right). $$
Conversely, we can check that this point is indeed an equilibrium.
The eigenvalues of the Jacobian matrix at this point are $ (2\beta(\sqrt{\beta+1}+1)^2)^{-1} $ times the roots of the polynomial
$$ x^3 +a x^2+ex+f ,$$
with
\begin{multline*}
0>f= -2 b^2 \beta^3 \left(\beta^2-1\right)(b \beta+b-2 d)\\
\left(\left(\sqrt{\beta+1}+5\right) \beta^2+4 \left(3 \sqrt{\beta+1}+5\right) \beta +16 \left(\sqrt{\beta+1}+1\right)\right) 
\end{multline*}
and 
\begin{multline*} 0<a= 6 b \beta + 6 b \beta^2 + 6 b \beta \sqrt{1 + \beta} + 
 2 b \beta^2 \sqrt{1 + \beta} - 4 \beta d \\
 - 2 \beta^2 d - 4 \beta \sqrt{1 + \beta} d .\end{multline*}
We deduce that there are two negative and one positive eigenvalues. This equilibrium is thus unstable.

As a conclusion, in the dominant case, the only stable fixed points are $(\zeta,0,0)$ and $(0,0,\zeta)$.
\end{proof}

\section{Stability of the monomorphic fixed points in the system with migration}

As the systems \eqref{systcompletcodom} and \eqref{systcompletdom} are complex, we are not able to derive an analytical form for all the fixed points. As a consequence in the next two subsections we focus 
on the fixed points which are monomorphic in one patch, $(\zeta,0,0,\zeta,0,0)$, $(0,0,\zeta,0,0,\zeta)$,
$(\zeta,0,0,0,0,\zeta)$, and $(0,0,\zeta,\zeta,0,0)$. 

\subsection{First codominant case} \label{cod1_mig}

We will prove that for the first codominant model the monomorphic fixed points have the following properties:
\begin{lem}\label{lemmecod2}
 There are four fixed points which are monomorphic in every patch:
 \begin{itemize}
  \item $(\zeta,0,0,\zeta,0,0)$ and $(0,0,\zeta,0,0,\zeta)$ are stable.
  \item $(\zeta,0,0,0,0,\zeta)$ and $(0,0,\zeta,\zeta,0,0)$ are stable if $2p<b\beta(\beta-1)$ and unstable if $2p>b\beta(\beta-1)$.
 \end{itemize}
\end{lem}

Since the alleles $a$ and $A$ are codominant, the two fixed points $(\zeta,0,0,\zeta,0,0)$ and $(0,0,\zeta,0,0,\zeta)$ have the 
same characteristics.  The eigenvalues of the 
Jacobian Matrix at these fixed points are
$$\left( -b \beta, -\frac{1}{2} b(\beta-1), d-b \beta, d-b \beta, -b \beta-2 p, -\frac{1}{2} (b (\beta-1)+2 p) \right).$$
They are all negative, and we conclude that $(\zeta,0,0,\zeta,0,0)$ and $(0,0,\zeta,0,0,\zeta)$ are two stable fixed points.

In the codominant case, the two fixed points $(\zeta,0,0,0,0,\zeta)$ and $(0,0,\zeta,\zeta,0,0)$ have the same properties since the 
alleles $a$ and $A$ are codominant. The eigenvalues of the 
Jacobian Matrix at these fixed points are
\begin{multline*}\bigg( d-b \beta, d-b \beta, \frac{1}{4} \left(-\sqrt{b} \sqrt{b \beta^2+2 b \beta+b+16 p}-3 b \beta+b\right),\\ 
\frac{1}{4} \left(\sqrt{b} \sqrt{b \beta^2+2 b \beta+b+16 p}-3 b \beta+b\right), \\
\frac{1}{4} \left(-\sqrt{b^2 \beta^2+2 b^2 \beta+b^2+4 b \beta p-12 b p+4 p^2}-3 b \beta+b-6 p\right), \\
\frac{1}{4} \left(\sqrt{b^2 \beta^2+2 b^2 \beta+b^2+4 b \beta p-12 b p+4 p^2}-3 b \beta+b-6 p\right) \bigg). \end{multline*}
All the eigenvalues, except the fourth one, are negative. The fourth eigenvalue is negative only if the migration parameter $p$ is small enough:
$$ \frac{1}{4} \left(\sqrt{b} \sqrt{b \beta^2+2 b \beta+b+16 p}-3 b \beta+b\right)<0 \Leftrightarrow 2p < b \beta (\beta -1) $$
Hence the eigenvalues $(\zeta,0,0,0,0,\zeta)$ and $(0,0,\zeta,\zeta,0,0)$ are stable only if $2p < b \beta (\beta -1)$.

\subsection{Second codominant case} \label{cod2_mig}

Lemma \ref{lemmecod2} is still valid for the second codominant case, although the eigenvalues of the fixed points are not the same:
$$ (-b \beta, b - b \beta, 
 b - b \beta, -b \beta + d, -b \beta + d, -b \beta - 2 p) $$
for the fixed points $(\zeta,0,0,\zeta,0,0)$ and $(0,0,\zeta,0,0,\zeta)$, and 
\begin{multline*} \Big(-b \beta,d-b \beta,d-b \beta,-b\beta+b-2 p,
\frac{1}{2} \left(-2 b \beta-\sqrt{b} \sqrt{b+8 p}+b\right),\\
\frac{1}{2} \left(-2 b \beta+\sqrt{b} \sqrt{b+8 p}+b\right)\Big) \end{multline*}
for the fixed points $(\zeta,0,0,0,0,\zeta)$ and $(0,0,\zeta,\zeta,0,0)$.

\subsection{Dominant case} \label{dom_mig}

In the dominant case, the two fixed points $(\zeta,0,0,\zeta,0,0)$ and $(0,0,\zeta,0,0,\zeta)$ are not symmetrical, 
as the two alleles $A$ and $a$ do not have the same role. We prove the following Lemma.
\begin{lem}
There are four fixed points which are monomorphic in every patch:
 \begin{itemize}
  \item $(0,0,\zeta,0,0,\zeta)$ (equilibrium with genotype $aa$ in both patches) is stable and there exists a neighborhood $\mathcal{V}^{aa}$ of $(0,0,\zeta,0,0,\zeta)$ such that any solution starting from $\mathcal{V}^{aa}$ converges exponentially fast to the equilibrium,
  \item $(\zeta,0,0,\zeta,0,0)$ (equilibrium with genotype $AA$ in both patches) is stable and there exists a neighborhood $\mathcal{V}^{AA}$ of $(\zeta,0,0,\zeta,0,0)$ such that any solution starting from $\mathcal{V}^{AA}$ converges to the equilibrium with a rate $t\mapsto \frac{1}{t}$,
  \item $(\zeta,0,0,0,0,\zeta)$ and $(0,0,\zeta,\zeta,0,0)$ are two equilibria whose Jacobian Matrices admit five negative eigenvalues and a null eigenvalue.
 \end{itemize}
\end{lem}
In the case of the two last equilibria, we were not able to prove theoretically their stability. Simulations of the solution seem to show that the stability depends on the values of the parameters.

\begin{proof}
\textbf{Equilibrium $(0,0,\zeta,0,0,\zeta)$}: The fixed point $(0,0,\zeta,0,0,\zeta)$ is the easiest to study, as all the eigenvalues of the Jacobian Matrix at this point are negative:
$$  (-b\beta,b-b\beta,-b\beta+d,-b\beta+d, -b\beta-2p,b-b\beta-2p ) . $$
Hence, the equilibrium $(0,0,{\zeta},0,0,{\zeta})$ is stable.\\
\\
\textbf{Equilibrium $(\zeta,0,0,\zeta,0,0)$}: The second fixed point, $(\zeta,0,0,\zeta,0,0)$, is more involved to study, since the Jacobian Matrix at this points admits four negative eigenvalues, and $0$ as an eigenvalue with 
multiplicity two:
$$  (0,0,-b\beta,-b\beta+d, -b\beta+d,-b\beta-2p ) . $$
As a consequence, we have to go further to get the stability of this fixed point.
The Jacobian Matrix at this fixed point is
\begin{equation}
\label{mat_jacoZ00Z00}
\footnotesize{\left(
\begin{array}{cccccc}
 -b \beta+d & -b \beta+d & -2 b \beta+d-p & 0 & 0 & p \\
 0 & 0 & 2 b & 0 & 0 & 0 \\
 0 & 0 & -b \beta-p & 0 & 0 & p \\
 0 & 0 & p & -b \beta+d & -b \beta+d & -2 b \beta+d-p \\
 0 & 0 & 0 & 0 & 0 & 2 b \\
 0 & 0 & p & 0 & 0 & -b \beta-p \\
\end{array}
\right)},
\end{equation}
and its eigensystem is
$$\left( 0; (0, 0, 0, -1, 1, 0)\right)$$
$$\left( 0; (-1, 1, 0, 0, 0, 0)\right)$$
\begingroup\makeatletter\def\f@size{9}\check@mathfonts
$$\left( -b\beta ; \left(\frac{
 2 b \beta ( \beta-1) + 2 d - \beta d}{\beta d}, -\frac{2}{\beta}, 1, \frac{
 2 b \beta ( \beta-1) + 2 d - \beta d}{\beta d}, -\frac{2}{\beta}, 1\right) \right)$$
 \endgroup
 $$\left(-b\beta+d;( 0, 0, 0, 1, 0, 0)\right)$$
 $$\left(-b\beta+d;(1, 0, 0, 0, 0, 0)\right)$$
\begingroup\makeatletter\def\f@size{8}\check@mathfonts
 \begin{multline*} \Big( -b\beta-2p ;\left(-\frac{(2 b^2 \beta (\beta-1) +  b d(2 - \beta) + 6 b \beta p - 2 d p + 
   4 p^2)}{ (d + 2 p)(b \beta + 2 p)},  \frac{2 b}{
 b \beta + 2 p}, \right. \\ \left.  -1, \frac{2 b \beta - d + 2 p}{d + 2 p} + 
   \frac{2 b(d-b \beta) }{(b \beta + 2 p) (d + 2 p)}, -\frac{2 b}{
  b \beta + 2 p}, 1\right)\Big).\end{multline*}\endgroup
Hence if we  make a translation such that the fixed point has coordinates $(0,0,0,0,0,0)$, and introduce a new basis via the following matrix:
\begin{equation} \label{change_basis} \footnotesize{\left(
\begin{array}{cccccc}
 0 & -1 & -\frac{-2 b \beta^2+2 b \beta+d \beta-2 d}{\beta d} & 0 & 1 & -\frac{2 \beta^2 b^2-2 \beta b^2-\beta d b+2 d b+6 \beta p b+4 p^2-2 d p}{(b \beta+2 p) (d+2 p)} \\
 0 & 1 & -\frac{2}{\beta} & 0 & 0 & \frac{2 b}{b \beta+2 p} \\
 0 & 0 & 1 & 0 & 0 & -1 \\
 -1 & 0 & -\frac{-2 b \beta^2+2 b \beta+d \beta-2 d}{\beta d} & 1 & 0 & \frac{2 b (d-b \beta)}{(b \beta+2 p) (d+2 p)}-\frac{-2 b \beta+d-2 p}{d+2 p} \\
 1 & 0 & -\frac{2}{\beta} & 0 & 0 & -\frac{2 b}{b \beta+2 p} \\
 0 & 0 & 1 & 0 & 0 & 1 \\
\end{array}
\right)} ,\end{equation}
we get that the dynamical system \eqref{systcompletdom} can be written
\begin{equation}\label{center_manifold}
 \begin{aligned}
  \dot X&=G_1(X,Y),\\
  \dot Y&= DY+G_2(X,Y).
 \end{aligned}
\end{equation}
Here for every $r\in \N$, $X\in C^r(\R,\R^2)$, $Y\in C^r(\R,\R^4)$, 
$G_1\in C^r(\R^6,\R^2)$, $G_1(0,0)=0$, $DG_1(0,0)=0$, $G_2\in C^{r}(\R^6,\R^4)$, 
$G_2(0,0)=0$, $DG_2(0,0)=0$, and $D$ is a diagonal matrix whose main diagonal is : $(-b\beta,-b\beta+d, -b\beta+d,-b\beta-2p)$.\\
According to Theorem 2.12.1 in \cite{perko2013differential}, there exists a real function $h$ defined in a neighbourhood of $(0,0)$ and with values in $\R^4$ 
such that $h(0)=0$ and $Dh(0)=0$, which defines the center manifold. That is to say: close to the null fixed point, $Y=h(X)$. The function $h$ 
satisfies
\begin{equation}
 \label{eq_h}
 Dh(X) \cdot G_1(X,h(X))=D\cdot h(X)+G_2(X,h(X)).
\end{equation}
We only need the first non null order of the function $h$ in a neighbourhood of $(0,0)$. Hence we are looking for a function which satisfies
\begin{equation}\label{eq_h1}
 h: \begin{pmatrix} x_1 \\ x_2 \end{pmatrix} \to  \begin{pmatrix} h_{11}\\h_{12}\\h_{13}\\h_{14} \end{pmatrix} x_1^2 + 
 \begin{pmatrix} h_{21}\\h_{22}\\h_{23}\\h_{24} \end{pmatrix} x_1 x_2 +\begin{pmatrix} h_{31}\\h_{32}\\h_{33}\\h_{34} \end{pmatrix} x_2^2 + 
 O\left(\left\Vert \begin{pmatrix} x_1\\x_2\end{pmatrix} \right\Vert^3 \right),
\end{equation}
where the parameters $(h_{ij}, i \in \{1,2,3\}, j \in \{1,2,3,4\})$ have to be computed.
To compute these latter, we inject the expression of $h$ \eqref{eq_h1} into \eqref{eq_h} and compute the second order of the Taylor series which is the 
first non null order of the series. First we notice that the second order on the left hand side of \eqref{eq_h} is null. As a consequence, we only have to 
develop the right hand side and identify the coefficients $h_{ij}$. We get
\begingroup\makeatletter\def\f@size{7}\check@mathfonts
\begin{equation*}
 h\begin{pmatrix} x_1 \\ x_2 \end{pmatrix} = 
 \begin{pmatrix}
 \frac{c }{8 (b \beta-d)} (x_1^2+x_2^2) \\ 
 \frac{b \beta c}{2d (-b \beta+d)^2 (d+2 p)} (-(d p+b (-1+\beta) (d+p)) x_1^2 +(b-b \beta+d) p x_2^2) \\ 
 \frac{b \beta c }{2 d (-b \beta+d)^2 (d+2 p)} ((b-b \beta+d) p x_1^2-(d p+b (-1+\beta) (d+p)) x_2^2)\\
 \frac{b \beta c }{8 (b \beta-d) (b \beta+2 p)} (x_1^2-x_2^2)
 \end{pmatrix} + O\left(\left\Vert \begin{pmatrix} x_1\\x_2\end{pmatrix} \right\Vert^3 \right).
\end{equation*}
\endgroup
From the first equation in \eqref{center_manifold} we know that in a neighbourhood of $(0,0)$, $\dot X=G_1(X,h(X))$. Hence
\begin{equation*}
 \left\{
 \begin{aligned}
 \dot x_1 &= \frac{b c }{2 (b \beta-d) (b \beta+2 p)}[-\beta (b(\beta-1)+2 p) x_1^2 +p (x_1^2+x_2^2)]\\
 &\qquad \qquad \qquad \qquad \qquad \qquad \qquad + O\left(\left\Vert ( x_1,x_2) \right\Vert^3 \right),\\
 \dot x_2 &= \frac{b c }{2 (b \beta - d) (b \beta + 2 p)}[p (x_1^2+x_2^2) -\beta (b(\beta-1)+2 p) x_2^2]
 \\ & \qquad \qquad \qquad \qquad \qquad \qquad \qquad  + O\left(\left\Vert ( x_1,x_2) \right\Vert^3 \right).
 \end{aligned}
\right.
\end{equation*}
But we notice that the solutions $(\tilde x_1,\tilde x_2)$ of
\begin{equation}
 \left\{
 \begin{aligned}
 \dot{\tilde{x}}_1 &= \frac{b c }{2 (b \beta-d) (b \beta+2 p)}[-\beta (b(\beta-1)+2 p) \tilde x_1^2 +p ( \tilde x_1^2+\tilde x_2^2)],\\
 \dot{\tilde{x}}_2 &= \frac{b c }{2 (b \beta - d) (b \beta + 2 p)}[p (\tilde x_1^2+\tilde x_2^2) -\beta (b(\beta-1)+2 p) \tilde x_2^2].
 \end{aligned}
\right.
\end{equation}
satifiy
$$
\dot{\tilde{x}}_1 + \dot{\tilde{x}}_2 = -\frac{b^2\beta(\beta-1) c }{2 (b \beta-d) (b \beta+2 p)}(\tilde x_1^2+\tilde x_2^2) .
$$
Moreover, using that $z_1^2+z_2^2 \leq (z_1+z_2)^2 \leq 2 (z_1^2+z_2^2)$ for any $(z_1,z_2) \in \R_+^2$, we get
\begin{multline*}
-\frac{b^2\beta(\beta-1) c }{2 (b \beta-d) (b \beta+2 p)}(\tilde x_1+\tilde x_2)^2 \\ \leq\dot{\tilde{x}}_1 + 
\dot{\tilde{x}}_2 \leq \\
-\frac{b^2\beta(\beta-1) c }{4 (b \beta-d) (b \beta+2 p)}(\tilde x_1+\tilde x_2)^2.
\end{multline*}
Hence $x_1+x_2$ has a constant sign and converges to $0$ with a rate 
$t\mapsto \frac{1}{t}$ in a neighbourhood of the fixed point $(0,0,0,0,0,0)$.
But by inversion of the matrix \eqref{change_basis}, we see that
$$x_1+x_2=\frac{2}{\beta}(z_{aa,1}+z_{aa,2})+z_{Aa,1}+z_{Aa,2}$$
This ends the proof of the stability of the fixed point $(\zeta,0,0,\zeta,0,0)$. \\
Notice that this method has been used recently in a similar system in \cite{bovier2015survival}. In this paper the authors were interested in a system of diploid individuals, 
with dominance 
but without sexual preference and migration. Their center manifold was of dimension one.\\
\\
\textbf{Equilibria $(\zeta,0,0,0,0,\zeta)$ and $(0,0,\zeta,\zeta,0,0)$}: Those two equilibria with different monomorphic populations in the two patches 
are symmetrical. We only prove the results on $(0,0,\zeta,\zeta,0,0)$. 
The eigenvalues of the Jacobian Matrix at this fixed point are
$$ (0, -b\beta+d, -b\beta+d), $$
as well as the roots of the polynomial:
\begin{multline*}
P(X)= X^3  + (b (3 \beta-1) + 3 p) X^2 
+b^3 \beta^2 ( \beta-1) + b^2 \beta p (3 \beta-2) \\+ 2 b p^2( \beta -1)  
+ ( b^2 \beta (3\beta-2) + 2 b p (3 \beta -1) + 2 p^2) X  \\
=: X^3 + A X^2 + EX+F.
\end{multline*}
We notice that the coefficients of $P$ are positive. Hence any real root is negative. This proves that all the roots are negative in the case where 
the three roots are real numbers. The second possibility is that $P$ admits one negative real root (denoted by $-\rho$) and two complex conjugate roots 
(denoted by $ \phi e^{i\theta} $ and $\phi e^{-i\theta}$, $\phi \in \R_+$, $\theta \in [0,2\pi)$).
We aim at showing that $\cos \theta < 0$, which will ensure that all the roots of $P$ have negative real parts.\\
With these notations, $P$ can be rewritten as follows
$$ P(X)=X^3+(\rho - 2\phi \cos \theta)X^2+\phi(\phi-2\rho \cos \theta)X+\rho \phi^2. $$
As the coefficients of $P$ are positive, we get
$$ \rho > 2\phi \cos \theta \text{ and } \phi > 2 \rho \cos \theta, $$
which yields
$$ 2 \cos \theta < \frac{\rho}{\phi}< \frac{1}{2\cos \theta},$$
and thus
$$ 4 \cos^2 \theta <1. $$
But we also have the following series of inequalities:
\begin{multline*}
0<b^3 \beta (2 + 7 (\beta-1) \beta) + 
 b^2 (2 + 7 \beta ( 3 \beta-2)) p + 4 b ( 5 \beta-1) p^2 + 6 p^3\\
= AE-2F = (\rho - 2\phi \cos \theta)\phi(\phi-2\rho \cos \theta)- 2\rho \phi^2\\ 
= \phi^2 \rho  (4\cos^2 \theta -1) -2 \phi \cos \theta (\phi^2+ \rho^2) \leq -2 \phi \cos \theta (\phi^2+ \rho^2).
\end{multline*}
This implies that $\cos \theta <0$, as expected.
Unfortunately, as we do not know explicitely the roots of $P$, we cannot apply Theorem 2.12.1 in \cite{perko2013differential} to get the behaviour of 
the solutions of \eqref{sys1patch0dom} in the center manifold and deduce the stability of the fixed point $(0,0,\zeta,\zeta,0,0)$.
\end{proof}

\section{Some properties of the system with migration}

\subsection{Polymorphic equilibria}

We were not able to determine all the polymorphic equilibria of the three dynamical systems governing the population dynamics in the two 
codominant and the dominant cases. Of course, thanks to our analysis of the cases without migration, we are able to give one unstable polymorphic 
equilibrium for each dynamical system:
\begin{itemize}
 \item $(\delta \xi, \xi, \delta \xi, \delta \xi, \xi, \delta \xi)$ for the first codominant model
 \item $(\delta' \xi', \xi', \delta' \xi', \delta' \xi', \xi', \delta' \xi')$ for the second codominant model
 \item $ \left(b \frac{\beta+1}{4c}-\frac{d}{2c}\right) \left( \frac{\sqrt{\beta+1}-1}{\sqrt{\beta+1}+1}, \frac{2}{\sqrt{\beta+1}+1},
1,\frac{\sqrt{\beta+1}-1}{\sqrt{\beta+1}+1}, \frac{2}{\sqrt{\beta+1}+1},
1\right) $ for the model with dominance
\end{itemize}

To give an idea of the complexity of the dynamical systems under study, for the parameters $b=2,d=1,c=0.5,\beta=1.1$ and $p=5$, Mathematica 
gives the numerical approximation of $9$ polymorphic equilibria for the second codominant model.

\subsection{A majority allele} \label{app_maj_allele}

Let us denote by $\bar{\alpha}$ the complement of $\alpha \in \{A,a\}$.
In both codominant models, if $z_{\alpha \alpha,i}(0)\geq z_{\bar{\alpha}\bar{\alpha},i}(0)$ for $i=1,2$, 
then $z_{\alpha \alpha,i}(t)\geq z_{\bar{\alpha}\bar{\alpha},i}(t)$ for $i=1,2$ and
for every positive $t$.
Indeed if the $(z_{\alpha \alpha',i}, \alpha \alpha' \in \{AA,Aa,aa\}, i \in \{1,2\})$ evolve according to the dynamical system \eqref{systcompletcodom}, 
we have
\begin{multline*}
 \dot{z}_{AA,1}-\dot{z}_{aa,1} = (z_{AA,1}-z_{aa,1}) \left( \beta b -d - c N_1- \frac{b(\beta-1)+p}{2}\frac{z_{Aa,1}}{N_1} \right)
\\ + \frac{p}{2}\frac{z_{Aa,2}}{N_2}(z_{AA,2}-z_{aa,2}),
\end{multline*}
and if the 
evolve according to the dynamical system \eqref{systcompletcodom2}, 
we have
\begin{multline*}
 \dot{z}_{AA,1}-\dot{z}_{aa,1} = (z_{AA,1}-z_{aa,1}) \left( \beta b -d - c N_1- (b(\beta-1)+p)\frac{z_{Aa,1}}{N_1} \right)
\\ + p\frac{z_{Aa,2}}{N_2}(z_{AA,2}-z_{aa,2}).
\end{multline*}
From these two equalities we get that the sets
$$ \mathcal{B}^=:= \{ z_{AA,1}=z_{aa,1}\} \cap \{ z_{AA,2}= z_{aa,2} \} $$
and
$$\mathcal{B}^>:= \{ z_{AA,1}>z_{aa,1}\} \cap \{ z_{AA,2}> z_{aa,2} \} $$
are invariant under the dynamical systems \eqref{systcompletcodom} and \eqref{systcompletcodom2}.

\section{The case of constant migration} \label{studyconstantpim}

For the sake of completeness, we study the case where all individuals have the same migration rate $p$, independent of the population state.
The population dynamics is thus solution to the following system of equations:
\begin{equation}
\label{eq_model_mig_constante}
\left\{
 \begin{aligned}
 \dot z_{AA,i}=&\frac{b }{N_i}\left( \beta z_{AA,i}^2+p_\beta(AA,Aa)z_{AA,i}z_{Aa,i}+\frac{p_\beta(Aa,Aa)}{4} z_{Aa,i}^2 \right) \\&-(d+cN_i)z_{AA,i}
  +p   \left(z_{AA,j}-z_{AA,i}\right)\\
 \dot z_{Aa,i}=&\frac{b}{N_i}\left( \frac{p_\beta(Aa,Aa)}{2} z_{Aa,i}^2+p_\beta(AA,Aa)z_{Aa,i}z_{AA,i}\right.\\&+p_\beta(aa,Aa)z_{Aa,i}z_{aa,i} 
  +2 z_{AA,i}z_{aa,i} \Big)-(d+cN_i)z_{Aa,i} \\
&+p \left(z_{Aa,j}-z_{Aa,i}\right)\\
 \dot z_{aa,i}=&\frac{b }{N_i}\left( \beta z_{aa,i}^2+p_\beta(aa,Aa)z_{aa,i}z_{Aa,i} +\frac{p_\beta(Aa,Aa)}{4} z_{Aa,i}^2 \right) \\
 &-(d+cN_i)z_{aa,i}+p \left(z_{aa,j}-z_{aa,i}\right).
\end{aligned}
\right. 
\end{equation}
We will show that in this case the equilibria are either null, or monomorphic, or have all their coordinates positive. More precisely
\begin{lem}\label{mono}
Let $(z_{\alpha\alpha',i}, (\alpha, \alpha') \in \{A,a\}^2, i \in \{1,2\} )$ be an equilibrium of Equation \eqref{eq_model_mig_constante}.
Then, either $z=0$, or all the coordinates of $z$ are positive, or $z$ has only two nonnull coordinates:
$$ z_{\alpha\alpha,1}=z_{\alpha\alpha,2}=\frac{b\beta-d}{c}, $$
for an $\alpha \in \{A,a\}$.
\end{lem}
Let us denote by $\bar{\alpha}$ the complementary of $\alpha$ in $\{A,a\}$.
The first step to get Lemma \ref{mono} is to prove the following Lemma.
\begin{lem} \label{lemmono}
Let $(z_{\alpha\alpha',i}, (\alpha, \alpha') \in \{A,a\}^2, i \in \{1,2\} )$ be an equilibrium of Equation \eqref{eq_model_mig_constante}.
Assume that there exists $(\alpha,i) \in \{A,a\}\times\{1,2\}$ such that $z_{\alpha\alpha,i}=0$.
Then $z=0$ or 
$$ z_{\bar{\alpha}\bar{\alpha},1}=z_{\bar{\alpha}\bar{\alpha},2}= \frac{b\beta-d}{c}, $$
and the other coordinates of $z$ are null.
\end{lem}
\begin{proof}
Let us first notice that $0$ is an equilibrium and that if the population size of one patch is null, the population size of the other patch is also null.
Thus we may assume in the rest of the proof without loss of generality that $z_{AA,1}=0$, that $N_1>0$ and $N_2>0$. Then the first equation in \eqref{eq_model_mig_constante} writes:
 $$ 0=\frac{b}{N_1}\frac{p_\beta(Aa,Aa)}{4} z_{Aa,1}^2 +pz_{AA,2} ,$$
 and hence,
 $$ z_{Aa,1}=z_{AA,2}=0. $$
 As $z_{AA,2}=0$, using the same reasoning we get that $z_{Aa,2}=0$. Now we are left with the last equation in \eqref{eq_model_mig_constante}, which can be written:
 \begin{equation} \label{zaa1} 0=( b\beta  -(d+cz_{aa,1}))z_{aa,1} +p (z_{aa,2}-z_{aa,1}). \end{equation}
 Similarly we get
 \begin{equation} \label{zaa2} 0=( b\beta  -(d+cz_{aa,2}))z_{aa,2} +p (z_{aa,1}-z_{aa,2}). \end{equation}
First, notice that a possible solution is $z_{aa,1}=z_{aa_2}=(b\beta-d)/c$. Assume that $z_{aa,1}\neq z_{aa,2}$. Then without loss of generality, 
we may assume that $z_{aa,2}>z_{aa,1}$. 
In this case, \eqref{zaa1} entails:
$$ b\beta  -(d+cz_{aa,1})<0 \Longleftrightarrow z_{aa,1}> \frac{b\beta -d}{c}, $$
hence
\begin{equation} \label{caszaa2plusgrand} z_{aa,2}>z_{aa,1}> \frac{b\beta -d}{c}. \end{equation}
Now let us add Equations \eqref{zaa1} and \eqref{zaa2}. We obtain
$$ ( b\beta  -(d+cz_{aa,1}))z_{aa,1} +( b\beta  -(d+cz_{aa,2}))z_{aa,2}=0 ,$$
or in others words
\begin{equation}\label{Ropposes} R[z_{aa,1}]+ R[z_{aa,2}]=0, \end{equation}
where $R$ is the polynomial function which at $X$ associates 
$$ R[X]=cX^2-(b\beta-d)X. $$
This polynomial function is negative on $I_1:=(0,(b\beta-d)/c)$ and positive on $I_2:=((b\beta-d)/c,\infty)$.
Thus \eqref{Ropposes} implies that one of the $z_{aa,i}$s' belongs to $I_1$, and the other one to $I_2$. This contradicts \eqref{caszaa2plusgrand}.
\end{proof}
The second and last step to get Lemma \ref{mono} is to notice that if $z$ is an equilibrium of \eqref{eq_model_mig_constante} and if $z_{Aa,1}=0$, 
then the second equation in \eqref{eq_model_mig_constante} writes:
$$ \frac{2b}{N_1} z_{AA,1}z_{aa,1}+p z_{Aa,2}=0. $$
This implies that $z_{Aa,2}=0$, and $z_{AA,1}=0$ or $z_{aa,1}=0$. Then we may apply Lemma \ref{lemmono}.
This ends the proof of Lemma \ref{mono}.
 

   \bibliographystyle{apalike} 
  \bibliography{biblio_speciation}




%
\end{document}